\documentclass[journal]{IEEEtran}
\usepackage{amsmath,amssymb,amsfonts}
\usepackage{algorithmic}
\usepackage{graphicx}
\usepackage{textcomp}
\usepackage{xcolor}
\usepackage{epstopdf}
\usepackage{mathrsfs}
\usepackage{pifont}
\usepackage{multirow}
\usepackage{algorithm}
\usepackage{algorithmic}
\usepackage{amsthm}
\usepackage{bbm}
\newtheorem{lemma}{\textbf{Lemma}}
\usepackage{bm}
\usepackage{subfigure}

\newtheorem{thm}{\textbf{Theorem}}
\newtheorem{mydef}{\textbf{Definition}}[section]
\newtheorem{proposition}{\textbf{Proposition}}
\usepackage{color}

\begin{document}

\title{How Should I Orchestrate Resources of My Slices for Bursty URLLC Service Provision?}
\author{Peng Yang, Xing Xi, Tony Q. S. Quek,~\IEEEmembership{Fellow,~IEEE}, Jingxuan Chen, Xianbin Cao,~\IEEEmembership{Senior Member,~IEEE},
        Dapeng Wu,~\IEEEmembership{Fellow,~IEEE}
\thanks{
P. Yang and T. Q. S. Quek are with the Information Systems Technology and Design, Singapore University of Technology and Design, 487372 Singapore (e-mail: \text{\{peng\_yang, tonyquek\}}@sutd.edu.sg).

X. Xi, J. Chen, and X. Cao are with the School of Electronic and Information Engineering, Beihang University, Beijing 100083, China, and also with the Key Laboratory of Advanced Technology, Near Space Information System (Beihang University), Ministry of Industry and Information Technology of China, Beijing 100083, China (e-mail: \text{\{xixing, chenjingxuan, xbcao\}}@buaa.edu.cn).

D. Wu is with the Department of Electrical and Computer Engineering, University of Florida, Gainesville FL 32611 USA (e-mail: dpwu@ieee.org).
}
}





\maketitle

\begin{abstract}
  Future wireless networks are convinced to provide flexible and cost-efficient services via exploiting network slicing techniques. However, it is challenging to configure slicing systems for bursty ultra-reliable and low latency communications (URLLC) service provision due to its stringent requirements on low packet blocking probability and low codeword decoding error probability. In this paper, we propose to orchestrate network resources for a slicing system to guarantee more reliable bursty URLLC transmission. We re-cut physical resource blocks and derive the minimum upper bound of bandwidth for URLLC transmission with a low packet blocking probability. We correlate coordinated multipoint beamforming with channel uses and derive the minimum upper bound of channel uses for URLLC transmission with a low codeword decoding error probability. Considering the agreement on converging diverse services onto shared infrastructures, we further investigate the network slicing for URLLC and enhanced mobile broadband (eMBB) service multiplexing. Particularly, we formulate the service multiplexing as an optimization problem, which is challenging to be mitigated due to requirements of future channel information and of tackling a two timescale issue. To address the challenges, we develop a resource optimization algorithm based on a sample average approximate technique and a distributed optimization method with provable performance guarantees.
\end{abstract}




\maketitle

\section{Introduction}
Future wireless networks are desired to provide diverse service requirements concerning throughput, latency, reliability, availability as well as operational requirements, e.g., energy efficiency and cost efficiency \cite{Matti2019Key,rost2017network}.
These service requirements are made by mobile networks and some novel application areas such as Industry 4.0, airborne communication, vehicular communication, and smart grid.

The International Telecommunication Union (ITU) has categorized
these services into three primary use cases: enhanced mobile broadband (eMBB), massive machine-type communications (mMTC), and ultra-reliable and low latency communications (URLLC) \cite{series2015imt}.
In order to provide cost-efficient solutions, it is agreed by some telecommunication organizations including the Third Generation
Partnership Project (3GPP) and the Next Generation Mobile
Network Alliance (NGMA), on the convergence of each use case onto a shared physical infrastructure instead of deploying individual network solution for each use case \cite{alliance2016description}.

To satisfy the requirement of reducing cost efficiency, the concept of network slicing has been proposed. The fundamental idea of network slicing is to logically isolate network resources and functions customized for specific requirements on a common physical infrastructure \cite{rost2017network}. A network slice as a virtual end-to-end (E2E) network for efficiently implementing resource isolation and increasing statistical multiplexing is self-contained with its virtual network resources, topology, traffic flow, and provisioning rules \cite{rost2017network, foukas2017network}.
Due to the significant role in constructing flexible and scalable future wireless networks, network slicing for mMTC, eMBB, and URLLC service (multiplexing) has received much attention from the academia \cite{albonda2019efficient,alsenwi2019embb,matera2018non}.

However, most of the current works do not study the impact of a time-varying channel on slice creation and benefits of exploiting advanced radio access techniques (RATs) in network slicing systems. For example, the actual channel may vary in short timescales (e.g., milliseconds) while slice creation may be conducted in relatively long timescales (e.g., minutes or hours). Therefore, network slicing needs to mitigate a multi-timescale issue. Additionally, the utilization of advance RATs (e.g., coordinated multipoint, CoMP) has been considered as a promising way of satisfying spectrum challenges and improving system throughput \cite{georgakopoulos2019coordination,maccartney2019millimeter}.

{
To tackle the multi-timescale issue of RAN slicing, the recent works in \cite{li2020hierarchical,ye2018dynamic} developed a software-defined networking (SDN)-based radio resource allocation framework. This framework could facilitate spectrum exploitation among different network slices in different time scales. Besides,} the work in \cite{tang2019service} developed a CoMP-based radio access network (RAN) slicing framework for eMBB and URLLC service multiplexing and proposed to tackle the multi-timescale issue of RAN slicing via an alternating direction method of multipliers (ADMM).
This work, however, assumed that URLLC traffic was uninterruptedly generated and ignored the significant bursty characteristic of URLLC traffic \cite{azari2019risk}. The bursty URLLC traffic (e.g., remote surgery
and remote robot control traffic) will further exacerbate the difficulty of slicing the RAN for URLLC involved service multiplexing from the following two aspects:
\begin{itemize}
    \item {\verb|Resource efficiency|}: one of the efficient proposals in future wireless communication networks to handle the uncertainty (including bursty) is to reserve network resources, which may waste a large number of valuable network resources. Therefore, it is important to develop resource orchestration schemes with high utilization for future networks, especially for some resource-constrained networks.
    \item {\verb|Immediate resource orchestration|}: bursty URLLC packets need to be immediately scheduled if there are available resources and the system utility can be maximized. Therefore, under the premise of improving resource efficiency, immediate resource orchestration schemes related to the number of flashing URLLC packets should be developed.
\end{itemize}

The difficulty motivates us to investigate the CoMP-enabled RAN slicing for bursty URLLC and eMBB service provision, and the primary contributions of this paper can be summarized as follows:
\begin{itemize}
  \item We re-cut physical resource blocks (PRBs) and derive the minimum upper bound of network bandwidth orchestrated for bursty URLLC traffic transmission to guarantee that the bursty URLLC packet blocking probability is of the order of a low value.
  \item After correlating CoMP beamforming with channel uses according to the network capacity result for the finite blocklength regime, we derive the minimum upper bound of channel uses required for transmitting a URLLC packet with a low codeword decoding error probability.
  \item We define eMBB and URLLC long-term slice utility functions and formulate the CoMP-enabled RAN slicing for bursty URLLC and eMBB service multiplexing as a resource optimization problem. The objective of the problem is to maximize the long-term total slice utility under constraints of total transmit power and network bandwidth. It is highly challenging to mitigate this problem due to the requirements of future channel information and of tackling a two timescale issue.
  \item To address the challenges, we propose a bandwidth and beamforming optimization algorithm. In this algorithm, we approximately transform the service multiplexing problem into a non-convex single timescale problem via a sample average approximate (SAA) technique. We exploit a distributed optimization method to mitigate the single timescale problem. Meanwhile, a semidefinite relaxation (SDR) scheme joint with a variable slack scheme is applied to transform the non-convex problem into a semidefinite programming (SDP) problem. We also perform theoretical analysis on the tightness and convergence of the proposed algorithm.
  \item At last, the performance of the proposed algorithm is validated through the comparison with the state-of-the-art algorithm.
\end{itemize}

The rest of the paper is organized as the following. In Section II, we review the related work. In Section III, we describe our system model and formulate the studied problem in Section IV. In Sections V and VI, we discuss the problem-solving method. Simulation results are given in Section VII, and this paper is concluded in Section VIII.

\section{Related work}
Enabling network slicing in 5G and beyond networks faces many challenges, in part owing to challenges in virtualizing and apportioning the RAN into several slices. To tackle these challenges, a rich body of previous works has been developed. In the following, we introduce some of the representatives on slice virtualization and resource apportionment.

In the research domain of slice virtualization, for example, a RAN slicing system for single RAT setting was developed to enable the dynamic virtualization of base stations (BSs) in \cite{foukas2017orion}. A control framework focusing on the balance of realistic traffic load and the deployment of virtual network functions was designed in \cite{ni2019end}. Based on {network functions virtualization (NFV)} technology, the works in \cite{landi2019provisioning, buyakar2018poster} proposed to automatically scale virtual network slices for content delivery (e.g., eMBB and mMTC traffic). Based on {SDN} and {NFV} technologies, slow startup and virtual Internet of things (IoT) network slices were created in \cite{wang2018poster} to meet different quality of service (QoS) requirements in IoT systems. To tackle the low-speed issue of constructing virtual network slices, a lightweight network slicing orchestration architecture was developed in \cite{li2019lightweight}.

In the research domain of resource apportionment, most of the literature focused on the resource abstraction and sharing. For instance, many recent works mapped resource sharing problems as the interaction between network resource providers and network slice brokers (or tenants). Scheduling mechanisms \cite{mandelli2019satisfying,marquez2018should}, game frameworks \cite{caballero2017network,caballero2018network,zheng2018statistical}, optimization frameworks \cite{alsenwi2019embb,matera2018non,leconte2018resource,d2018low,sciancalepore2017mobile,jiang2016network,bega2017optimising,sciancalepore2019rl,liu2019direct}, and artificial intelligence {(AI)}-based methods \cite{albonda2019efficient,gutterman2019ran} were then developed to help infrastructure providers improve profits (or utilities) and help tenants reap the benefits of resource sharing while guaranteeing their subscribers' service requirements.
Looking to resource abstraction, the work in \cite{ksentini2017toward} proposed a network slicing architecture featuring RAN resource abstraction, where a scheduling mechanism was crucial for abstracting network resources among slices. However, scheduling processes were not explored in more detail in this work. Using diverse resource abstraction types, an approach of virtualizing radio resources for multiple services was developed in \cite{chang2018radio} with the assumption that the traffic arrival rate of each slice equalled the number of requested radio resources.
Yet, few of the above works studied the benefit of slicing RAN equipped with advance RATs, e.g., CoMP.

Recently, there are some papers separately studying the CoMP without exploiting the network slicing \cite{ali2018coordinated,michaloliakos2017asynchronously, michaloliakos2016joint,cha2017coordinated,wu2019latency,navarro2018stochastic}. The fundamental principle of CoMP is similar to that of a distributed multiple-input multiple-output (MIMO) system, where CoMP cells act as a distributed antenna array under a virtual BS in the MIMO system \cite{georgakopoulos2019coordination,ali2018coordinated}.
In \cite{michaloliakos2017asynchronously,michaloliakos2016joint}, a CoMP architecture coupled with a user-beam selection scheme aiming at achieving high-performance gains without generating high overhead were developed, where all beams were assumed to transmit at the same power level. The work in \cite{cha2017coordinated} discussed the frequent inter-beam handover issue, which was caused by covering high-speed moving devices, in a CoMP-enabled mobile communication system with a single BS.
To improve ground users' QoS, fronthual bandwidth allocation and CoMP were jointly optimized in \cite{wu2019latency} without considering the impact of a time-varying channel on the scheme of bandwidth allocation. Some measurement-based studies on CoMP to mitigate user outage and improve network reliability were conducted in \cite{maccartney2019millimeter,navarro2018stochastic}, respectively. Some papers investigated the problem of CoMP-enabled network slicing \cite{tang2019service,yang2020multicast}; for instance, the work in \cite{yang2020multicast} designed a prototype for the CoMP-enabled RAN slicing system incorporating with multicast eMBB and bursty URLLC traffic.


\section{System model}
We consider a CoMP-enabled RAN slicing system for URLLC and eMBB multiplexing service provision. {The CoMP technique is explored because it can significantly improve the service quality of URLLC and eMBB via spatial diversity \cite{khoshnevisan20195g} and is easy and cost-efficient for the RAN slicing implementation \cite{yang2020multicast}}. In this system, the time is discretized and partitioned into time slots and minislots, and a time slot includes $T$ minislots. The URLLC deadlines are within a single minislot.
There are $N^e$ ground eMBB user equipments (UEs), $N^u$ ground URLLC UEs and $J$ {remote radio heads (RRHs), which connect to a baseband unit (BBU)}. The eMBB UE set and the URLLC UE set are denoted as ${\mathcal I}^e=\{1, \ldots, N^e\}$, ${\mathcal I}^u = \{1,\ldots,N^u\}$, respectively. We assume that eMBB and URLLC UEs are randomly distributed in a considered communication area, and RRHs are regularly deployed. Besides, each RRH is assumed to be equipped with $K$ antennas, and each UE is equipped with an antenna. All RRHs cooperate to transmit signals to a UE such that the signal-to-noise ratio (SNR) of it can be significantly enhanced\footnote{This paper exploits the optimization of transmit beamformers, and the issues of beam alignment and beam selection are out of the scope of this paper.}. Meanwhile, a flexible frequency division multiple access (FDMA) technique is exploited to achieve the inter-slice and intra-slice interference isolation.


\subsection{RAN slicing system}
Fig. \ref{fig:fig_eMBB_single_slice} shows an architecture of a RAN slicing system adopted in this paper, which consists of four parts: end UEs, RAN coordinator (RAN-C), network slice management, and network providers. At the beginning of each time slot, the RAN-C will decide whether to accept or reject the received slice requests for serving end eMBB and URLLC UEs after checking the available resource information (e.g., PRBs and transmit power) and computing. If a slice request can be accepted, network slice management will be responsible for creating or activating corresponding types of virtual slices, the process of which is time-costly and usually in a timescale of minutes to hours.
Next, if a slice request admission arrives, network providers will find the optimal servers and paths to place virtual network functions to satisfy the required E2E service of the slice.
\begin{figure}[!t]
\centering
\includegraphics[width=2.9in]{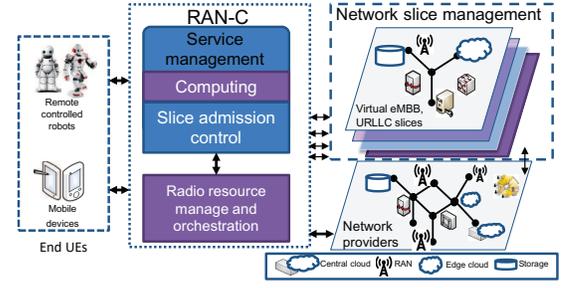}
\caption{The architecture of a RAN slicing system.}
\label{fig:fig_eMBB_single_slice}
\end{figure}

On the other hand, at the beginning of each minislot, coordinated RRHs will generate beamformers matching time-varying channels for each accepted slice.
In this RAN slicing system, we consider two types of slices, i.e., multicast eMBB slices and unicast URLLC slices\footnote{{We consider the multicast transmission for eMBB slices since the multicast transmission is envisioned to be a popular transmission scheme in the eMBB use case \cite{tang2019service}. As different URLLC UEs may require diverse data content (e.g., control signals in autonomous driving), unicast URLLC slices are considered.}}. The set of eMBB slices is denoted by {${\mathcal S}^e=\{1,2,\ldots,|{\mathcal S}^e|\}$ with $|{\mathcal S}^e|$ being the number of eMBB slices},
and the URLLC slice set is denoted by ${\mathcal S}^u =$ $\{1,2,\ldots,|{\mathcal S}^u|\}$, where $|{\mathcal S}^u|$ is the number of URLLC slices.

\subsection{eMBB slice model}
According to the above mentioned concept of a network slice (especially from the perspective of the QoS requirement of a slice), we can define an eMBB network slice request as the following.
\begin{mydef}[Multicast eMBB slice request]\label{embb_slice_res_def}
A multicast eMBB slice request can be characterized as a tuple $\{I_s^e, C_s^{th}\}$ for any multicast slice $s \in {\mathcal S}^e$, where $I_s^e$ is the number of multicast eMBB UEs in $s$ and $C_s^{th}$ is the data rate requirement of each UE in $s$.
\end{mydef}

In this definition, eMBB UEs are partitioned into $|{\mathcal S}^e|$ groups according to the data rate requirement of a UE. UEs in the same slice have the same data rate requirement. The slice request of each group of eMBB UEs will always be admitted by the RAN-C in this paper, and coordinated beamformers and PRBs will be effectively configured to accommodate data rate requirements of all eMBB UEs by way of the multicast transmission.

\subsection{URLLC slice model}
Similar to the definition of an eMBB slice request, we define the unicast URLLC slice request as follows.
\begin{mydef}[Unicast URLLC slice request]\label{urllc_slice_definition}
A unicast URLLC slice request can be characterized as four tuples $\{I_s^u, D_s, \alpha, \beta\}$ for any unicast slice $s \in {\mathcal S}^u$, where $I_s^u$ denotes the number of unicast URLLC UEs of $s$, $D_s$ represents the latency requirement of each UE in $s$, $\alpha$ and $\beta$ are denoted as the data packet blocking probability and the packet decoding error probability of each URLLC UE, respectively.
\end{mydef}

In this definition, URLLC UEs are classified into $|{\mathcal S}^u|$ clusters according to the latency requirement of each UE. Owing to the ultra-low latency requirement URLLC traffic should be immediately scheduled upon arrival; thus, URLLC slice requests will always be accepted by the RAN-C in this paper\footnote{We consider that the RAN-C can always accept eMBB and URLLC service requests here. However, whether these service requests can always be accommodated may be determined by both the QoS requirements of eMBB and URLLC UEs and the system service capability in practice.}. Then, coordinated beamformers will be correspondingly generated to cover UEs by way of the unicast transmission at the beginning of each minislot.

\section{Problem formulation}
On the basis of the above system model, this section aims to formulate the problem of RAN slicing for URLLC and eMBB multiplexing service provision.

\subsection{QoS requirements}
\subsubsection{Data rate requirements of eMBB UEs}
The generated transmit beamformers for UEs of slice $s$ ($s \in {\mathcal S}^e$) on RRH $j$ at minislot $t$ is denoted by ${\bm v}_{j,s}(t) \in {\mathbb C}^K$. The channel coefficient between RRH $j$ and eMBB UE $i$ of $s$ at minislot $t$ is denoted by ${\bm h}_{ij,s}(t) \in {\mathbb C}^K$, which does not greatly change in each minislot. Suppose that the instantaneous channel coefficient ${\bm h}_{ij,s}(t)$ can be effectively estimated by exploiting some machine learning methods \cite{yuan2020machine} at the beginning of minislot $t$ and the channel fading process is ergodic over a time slot for each $(i,j)$ pair. The SNR received at UE $i$ of slice $s$ at $t$ can then be written as
\begin{equation}\label{eq:eMBB_snr}
    SNR_{i,s}^e(t) = \frac{{|\sum\nolimits_{j \in {\mathcal J}} {{{\bm h}_{ij,s}^{\rm H}}{{(t)}}{{\bm v}_{j,s}}(t)} |^2}}{{\sigma _{i,s}^2}}, {\rm for} \text{ } {\rm all} \text{ } {\rm s}\in {\mathcal S}^e,i\in {\mathcal I}_s^e,
\end{equation}
where $\sigma_{i,s}^2$ denotes the noise power, ${\mathcal I}_s^e = \{1,\ldots,I_s^e\}$ is the set of eMBB UEs of $s$. Since the multicast transmission and flexible FDMA mechanism are exploited the interference is not involved.

According to Shannon formula, the achievable data rate $\gamma _{i,s}^e(t)$ of UE $i$ of slice $s$ at $t$ can be expressed as
\begin{equation}\label{eq:gamma_ise}
\gamma _{i,s}^e(t) = \omega_s^e(t){\log _2}(1 + SNR_{i,s}^e(t)),{\rm for} \text{ } {\rm all} \text{ } {\rm s}\in {\mathcal S}^e,i\in {\mathcal I}_s^e,
\end{equation}
where $\omega_{s}^e(t)$ denotes the bandwidth allocated to $s$ at $t$.

The following rate-related condition should be satisfied to admit the eMBB slice request
\begin{equation}\label{eq:eMBB_QoS}
    \gamma _{i,s}^e(t) \ge C_s^{th},{\rm for} \text{ } {\rm all} \text{ } {\rm s}\in {\mathcal S}^e,i\in {\mathcal I}_s^e.
\end{equation}

\subsubsection{QoS requirements of URLLC UEs}
As is known, it is challenging to design a RAN slicing system to support the transmission of URLLC traffic owing to URLLC UEs' stringent QoS requirements. What makes the issue more difficult is that URLLC traffic may be bursty. Bursty URLLC traffic, which may cause severe packet blocking, may significantly degrade the system performance of RAN slicing when URLLC slices are not well configured. To understand the characteristic of bursty URLLC traffic and mitigate the effect of bursty URLLC traffic on RAN slicing, we will address the following two questions.
\begin{itemize}
    \item \emph{How to model bursty URLLC traffic?}
    \item \emph{What schemes can be developed for the RAN slicing system such that the URLLC packet blocking probability can be significantly reduced?} 
\end{itemize}

During a time slot, bursty URLLC data packets destined to UEs of each URLLC slice and aggregated at the RAN-C are modelled as a compound Poisson process \cite{becchi2008poisson}, where arrivals happen in bursts (or batches, i.e., several arrivals can happen at the same instant) and the inter-batch duration is independent and exponentially distributed.
The vector of URLLC packet arrival rates is denoted by ${\bm \lambda} = \{\lambda_1,\ldots,\lambda_s, \ldots,\lambda_{|{\mathcal S}^u|}\}$, where $\lambda_s = \lambda_{b,s}/\lambda_{a,s} $ is a constant and represents the average arrival rate of packets destined to UEs of slice $s$ during a unit of time, $\lambda_{a,s}$ is the average inter-batch time interval, and $\lambda_{b,s}$ is the average number of arrivals in a batch.

On the basis of the URLLC traffic model, we next discuss how to reduce the URLLC packet blocking probability via re-cutting PRBs.
To satisfy QoS requirements of URLLC UEs, a portion of PRBs should be allocated to them. In the RAN slicing system, a URLLC UE $i$ of $s$ will be allocated a block of network bandwidth of size $\omega_{i,s}^{u}(t)$ for a period of time $d_{s}$ at minislot $t$. Since URLLC packets in $s$ have the deadline of $D_s$ seconds for E2E transmission latency, we shall always choose $d_{s} \le D_s$. Besides, a packet destined to the UE $i$ will be coded before sending out to improve the reliability\footnote{\emph{Packet} and \emph{codeword} are different terminologies. \emph{Packet} is a terminology used in the network layer (measuring in bits), and \emph{codeword} is the one adopted in the physical layer (measuring in symbols).}; and the transmission of a codeword needs $r_{i,s}^u(t)$ channel uses that measure the speed and capacity of a specific information channel.
The channel use, bandwidth and transmission latency are related by $r_{i,s}^{u}(t) = \kappa \omega_{i,s}^{u}(t) d_{s}$ \cite{anand2018resource}, where $\kappa $ is a constant representing the number of channel uses per unit time per unit bandwidth of the FDMA frame structure and numerology. We denote the channel use set of URLLC UEs as ${\bm r(t)} = \{{\bm r}_1(t), \ldots, {\bm r}_s(t), \ldots, {\bm r}_{|{\mathcal S}^u|}(t)\}$ with ${\bm r}_s(t) = \{r_{1,s}^u(t),\ldots,r_{I_s^u,s}^u(t)\}$.

Let us model the aggregation and departure of URLLC packets in the RAN-C as an $M/M/W^u$ queueing system, {which has higher key performance indicators than some other queueing systems like $M/M/1$}, with finite bandwidth $W^u$ and arrival data rate $\bm \lambda$.
However, the $M/M/W^u$ queueing system itself may not capture deadline-constrained jobs. To tackle this issue, we introduce a queueing probability concept into the queueing system.
Once entering into a queue, URLLC packets should be immediately served. Nevertheless, due to stochastic variations in the packet arrival process, there may not be enough spare bandwidth to serve new arrivals occasionally, and new arrivals should wait in the queue and may be blocked.
Then, we should take actions to reduce the blocking probability of URLLC packets.
Effectively re-cutting the PRBs is a way of significantly reducing the URLLC packet blocking probability.
Denote $p_{b}({\bm \omega}^u(t), {\bm \lambda},{\bm d},W^u(t))$ as the blocking probability experienced by arrival URLLC packets at minislot $t$ where ${\bm \omega}^u(t) = \{\omega_{1,1}^u$ $(t), \ldots, \omega_{i,s}^u(t), \ldots, \omega_{I_{|{\mathcal S}^u|}^u,|{\mathcal S}^u|}^u(t)\}$ and ${\bm d} = \{d_1, \ldots, d_s \ldots, d_{|{\mathcal S}^u|}\}$. The following theorem provides us a clue of re-cutting PRBs for URLLC packet transmission.

\begin{thm}\label{thm:time_frequency_plane_design}
At any minislot $t$, for the given ${\bm \omega}^u(t)$, $\bm d$, and a positive integer $q$, define ${\tilde {\bm \omega}}^u(t) = \{\omega_{1,1}^u(t), \ldots, \omega_{i,s}^u(t)/q, \ldots, \omega_{I_{|{\mathcal S}^u|}^u,|{\mathcal S}^u|}^u$ $(t)\}$ and ${\tilde{\bm d}} = \{d_1, \ldots, q d_s, \ldots, d_{|{\mathcal S}^u|}\}$. If $\lambda_s d_s < 1$, then there exists a value $\tilde W^u(t)$ such that for $W^u(t) > \tilde{W}^u(t)$ we have $p_{b}({\bm \omega}^u(t), {\bm \lambda},{\bm d},W^u(t)) \ge p_{b}({\tilde {\bm \omega}}^u(t), {\bm \lambda},{\tilde{\bm d}},W^u(t))$ \cite{anand2018resource}.
\end{thm}

This theorem tells us that if we shorten the packet latency, then fewer resource blocks will be available in the frequency plane, which will definitely cause more severe queueing effect and will significantly increase the packet blocking probability. If we narrow resource blocks in the frequency plane, then more concurrent transmissions are available, which is beneficial for decreasing the packet blocking probability.

Therefore, we should scale up $d_{s}$ and select $d_{s}$ and $\omega_{i,s}^u(t)$ for any URLLC slice $s$ at any minislot according to the following equation
\begin{equation}\label{eq:omega_s}
   d_{{s}}  = D_s \ {\rm and} \ \omega_{{i,s}}^{u}(t) = \frac{{r_{{i,s}}^{u}(t)}}{{\kappa D_s}}, {{\rm for} \text{ } {\rm all} \text{ } i \in {\mathcal I}_s^u}, s \in {\mathcal S}^u.
\end{equation}

With (\ref{eq:omega_s}), a square-root staffing rule \cite{harchol2013performance} can be exploited to derive the minimum upper bound of bandwidth allocated to URLLC slices such that the QoS requirements URLLC UEs can be satisfied. The following lemma presents the bound.
\begin{lemma}\label{lem:upper_bound_urllc_bandwidth}
    At any minislot $t$, for a given $M/M/W^u$ queueing system with packet arrival rates $\bm \lambda$ and packet transmit speeds $\{{\kappa/r_{i,s}^u(t)}\}$, let $W^u({\bm r}(t))$ denote the minimum upper bound of bandwidth allocated to all URLLC slices to ensure that $P_Q^{M/M/W^u}\le \varsigma$ and $p_b({\bm \omega}^u(t), {\bm \lambda}, {\bm D}, W^u({\bm r}(t)))$ is of the order of $\alpha$, where $P_Q^{M/M/W^u}$ represents the queueing probability. If $\varsigma > \alpha$, then
    \begin{equation}\label{eq:URLLC_bandwidth}
        W^u(\bm r(t)) \ge A({\bm r}(t)) + c(\varsigma,\alpha)\sqrt{B({\bm r}(t))},
    \end{equation}
    where $A({\bm r}(t)) = \sum\nolimits_{s \in {\mathcal S}^u } {\sum\nolimits_{i \in {\mathcal I}_s^u} {{\lambda _{s}}\frac{{{{r_{i,s}^u(t)}}}}{\kappa}} }$, $B({\bm r}(t)) = \sum\nolimits_{s \in {\mathcal S}^u} {\sum\nolimits_{i \in {\mathcal I}_s^u} {{\lambda _{s}}\frac{{r_{i,s}^u(t)}^2}{{{\kappa ^2}{D_s}}}} }$, ${\bm D} = \{D_1, \ldots, D_s\}$, and
    \begin{equation}\label{eq:calculate_c}
        c(\varsigma,\alpha) = \frac{{{\alpha} - \varsigma {\alpha}}}{{\varsigma  - {\alpha}}}\sqrt {\frac{{\sum\nolimits_{s \in {{\mathcal S}^u}} {I_s^u\lambda _s^2D_s^2} }}{{\mathop {\min }\nolimits_{s \in {{\mathcal S}^u}} \{ {\lambda _s}{D_s}\} }}}.
    \end{equation}
\end{lemma}
\begin{proof}
Please refer to Appendix A.
\end{proof}

In (\ref{eq:URLLC_bandwidth}), the first summation item denotes the mean value of the bandwidth allocated to URLLC slices, and the second summation item can be regarded as the redundant bandwidth allocated to mitigate the impact of stochastic variations in the arrival process.

We next discuss the URLLC capacity and channel uses.
For URLLC slice $s \in {\mathcal S}^u$, let $\bm g_{ij,s}(t) \in {\mathbb C}^K$ be the transmit beamformer to UE $i$ from RRH $j$ at $t$, ${{\bm h}_{ij,s}}{{(t)}}$ is the corresponding channel coefficient,
the corresponding SNR received at UE $i$ can then be expressed as
\begin{equation}\label{eq:URLLC_snr}
    SNR_{i,s}^u(t) = \frac{{|\sum\nolimits_{j \in {\mathcal J}} {{{\bm h}_{ij,s}^{\rm H}}{{(t)}}{{\bm g}_{ij,s}}(t)} {|^2}}}{{\phi \sigma _{i,s}^2}}, {{\rm for} \text{ } {\rm all} \text{ } i \in {\mathcal I}_s^u}, s \in {\mathcal S}^u.
\end{equation}

The perception of channel state information (CSI) or channel fading distribution may require the signal exchange before transmission, which entails extra transmit latency and potential reliability loss as well. Therefore it may be impossible to obtain perfect CSI for URLLC service provision, and a constant $\phi > 1$ is involved in (\ref{eq:URLLC_snr}) to model the SNR loss for URLLC traffic transmission \cite{liu2014energy}. Meanwhile, interference signals are not included in (\ref{eq:URLLC_snr}) as a flexible FDMA mechanism is exploited.

On the other hand, owing to the stringent low latency requirements, URLLC packets typically have very short blocklength. We therefore utilize the capacity result for the finite blocklength regime in \cite{polyanskiy2010channel,schiessl2015delay} to calculate the URLLC capacity rather than the Shannon formula that cannot effectively capture the reliability of packet transmission. Particularly, for each UE $i$ in $s \in {\mathcal S}^u$, the number of information bits $L_{i,s}^u(t)$ of a URLLC packet that is transmitted at $t$ with a codeword decoding error probability of the order of $\beta$ in $r_{i,s}^u(t)$ channel uses can be calculated by {\cite{polyanskiy2010channel,schiessl2015delay}}
\begin{equation}\label{eq:URLLC_bit_length}
L_{i,s}^u(t) \approx r_{i,s}^u(t)C(SNR_{i,s}^u(t)) -
{Q^{ - 1}}(\beta )\sqrt {r_{i,s}^u(t)V(SNR_{i,s}^u(t))},
\end{equation}
where $C(SNR_{i,s}^u(t)) = \log_2(1 + SNR_{i,s}^u(t) )$ is the AWGN channel capacity per Hz, $V(SNR_{i,s}^u(t)) = \ln^2 2\left( {1 - \frac{1}{{{{(1 + SNR_{i,s}^u(t))}^2}}}} \right)$ is the channel dispersion.

The expression of (\ref{eq:URLLC_bit_length}) is complicated; yet, the following lemma gives the approximate expression of channel uses in terms of codeword decoding error probability $\beta$ and SNR.
\begin{lemma}\label{lem:lemma_channel_use}
    For any UE $i$ in $s \in {\mathcal S}^u$, the required channel uses $r_{i,s}^u(t)$ of transmitting a URLLC packet of size of $L_{i,s}^u(t)$ to UE $i$ can be approximated as
    \begin{equation}\label{eq:URLLC_channel_use}
        \begin{array}{l}
        r_{i,s}^u(t) \le \frac{{L_{i,s}^u(t)}}{{C(SNR_{i,s}^u(t))}} + \frac{{{{({Q^{ - 1}}(\beta))}^2}}}{{2{{(C(SNR_{i,s}^u(t)))}^2}}}\\
         \qquad \quad + \frac{{{{({Q^{ - 1}}(\beta))}^2}}}{{2{{(C(SNR_{i,s}^u(t)))}^2}}}\sqrt {1 + \frac{{4L_{i,s}^u(t)C(SNR_{i,s}^u(t))}}{{{{({Q^{ - 1}}(\beta))}^2}}}}.
        \end{array}
    \end{equation}
\end{lemma}
\begin{proof}
Please refer to Appendix B.
\end{proof}

\subsection{Physical resource constraints}
As each RRH has a limitation on the maximum transmit power $E_j$ $(j \in {\mathcal J})$, we can obtain the following power constraint
\begin{equation}\label{eq:RRH_energy}
   \sum\limits_{s \in {{\mathcal S}^e}} {{{\bm v}_{j,s}^{\rm H}}{{(t)}}{{\bm v}_{j,s}}(t)}  + \sum\limits_{s \in {\mathcal S}^u} {\sum\limits_{i \in {\mathcal I}_s^u} {{{\bm g}_{ij,s}^{\rm H}}{{(t)}}{{\bm g}_{ij,s}}(t)} }  \le {E_j},
\end{equation}
{where the first item on the left-hand-side (LHS) of (\ref{eq:RRH_energy}) is the total transmit power consumption of RRHs for serving eMBB UEs and the second item for serving URLLC UEs.}

Besides, since the multicast eMBB and the unicast URLLC service provisions are considered, and network bandwidth resources allocated to eMBB and URLLC slices are separated in the frequency plane, the network bandwidth constraint can be written as
\begin{equation}\label{eq:total_bandwidth}
   \sum\limits_{s \in {{\mathcal S}^e}} {\omega_s^e(\bar t)}  + W^u(\bm r(t))  \le W,
\end{equation}
where ${\omega_s^e(\bar t)}$ represents the bandwidth allocated to eMBB slice $s \in {\mathcal S}^e$ over a time slot, $W$ denotes the maximum network bandwidth.

\subsection{Long-term utility function design}
We next discuss the design of the objective function of service multiplexing. To achieve the maximum utility of service multiplexing, utilities of eMBB and URLLC service provisions should be maximized simultaneously. In this paper, we leverage a key performance indicator, i.e., energy efficiency, which is popularly exploited in resource allocation problems, to model the utility.

On the one hand, as network states of any two adjacent slots can be seen as independent in the time-discrete RAN slicing system, we focus on the problem formulation in a time slot of a duration of $T$.
On the other hand, during a time slot, channel coefficients followed by the beamforming may change over minislots; as a result, time-varying utility functions in terms of channel coefficients and beamforming should be designed. {Besides, as the CoMP technique can significantly improve the SNR of UEs, we define the following two utility functions}.
\begin{mydef}[eMBB long-term utility function]\label{def:mydef_eMBB_utility}
Over a time slot, the eMBB long-term utility is defined as the time average energy efficiency of serving all eMBB UEs, which is calculated as
\begin{equation}\label{eMBB_long_term_utility}
    \begin{array}{l}
{{\bar U}^e} = \frac{1}{T}\sum\limits_{t = 1}^T {{U^e}(t)}  = \frac{1}{T}\sum\limits_{t = 1}^T {\sum\limits_{s \in {{\mathcal S}^e}} {u_s^e({\bm v_{j,s}}(t))} } \\
 = \frac{1}{T}\sum\limits_{t = 1}^T {\sum\limits_{s \in {{\mathcal S}^e}} {\left( {\sum\limits_{i \in {\mathcal I}_s^e} {SNR_{i,s}^e(t)}  - \eta \sum\limits_{j \in {\mathcal J}} {\bm v_{j,s}^{\rm{H}}(t){\bm v_{j,s}}(t)} } \right)} },
\end{array}
    \end{equation}
where $\eta$ is an energy efficiency coefficient reflecting the tradeoff between the {total transmit} power consumption {of RRHs} and the benefit {quantified by the total SNR of eMBB UEs}.
\end{mydef}
\begin{mydef}[URLLC long-term utility function]\label{eq:mydef_urllc_long_term}
Over a time slot, the URLLC long-term utility is defined as the time average energy efficiency of serving all URLLC UEs, which can be calculated as
\begin{equation}\label{URLLC_long_term_utility}
        \begin{array}{l}
        {{\bar U}^u} = \frac{1}{T}\sum\limits_{t = 1}^T {{U^u}(t)}  = \frac{1}{T}\sum\limits_{t = 1}^T {\sum\limits_{s \in {{\mathcal S}^u}} {u_s^u({\bm g_{ij,s}}(t))} } \\
         = \frac{1}{T}\sum\limits_{t = 1}^T {\sum\limits_{s \in {{\mathcal S}^u}} {( {\sum\limits_{i \in {\mathcal I}_s^u} {SNR_{i,s}^u(t)}  - \eta \sum\limits_{j \in {\mathcal J}} {\sum\limits_{i \in {\mathcal I}_s^u} {\bm g_{ij,s}^{\rm{H}}(t){\bm g_{ij,s}}(t)} } } )} }.
        \end{array}
    \end{equation}
\end{mydef}

With the above description, we can formulate the problem of RAN slicing for bursty URLLC and eMBB service multiplexing as follows
\begin{subequations}\label{eq:original_problem}
\begin{alignat}{2}
& \mathop {{\rm{maximize}}}\limits_{\{ \omega _s^e(\bar t),{\bm v_s}(t),{\bm g_{i,s}}(t)\}}  {\bar U = {\bar U}^e}  + {\hat \rho} {{\bar U}^u} \\
& {\rm s.t:} \text{ } \rm {constraints \text{ } (\ref{eq:eMBB_QoS}),(\ref{eq:RRH_energy}),(\ref{eq:total_bandwidth}) \text{ } are \text{ } satisfied,} \\
& \omega _s^e(\bar t) \ge 0, \forall s \in {\mathcal S}^e.
\end{alignat}
\end{subequations}
where $\hat \rho$ is a slice priority coefficient representing the priority of serving inter-slices, $\bar U$ denotes the long-term total slice utility, beamformers ${\bm v}_s(t) = [{\bm v}_{1,s}(t); \ldots; {\bm v}_{J,s}(t)] \in {\mathbb C}^{JK \times 1}$, and ${\bm g}_{i,s}(t) = [{\bm g}_{i1,s}(t);\ldots;$ ${\bm g}_{iJ,s}(t)] \in {\mathbb C}^{JK \times 1}$.

The mitigation of (\ref{eq:original_problem}) is highly challenging mainly because: {\verb|future channel information is needed|}: the optimization should be conducted at the beginning of the time slot; yet the objective function needs to be exactly computed according to channel information during the time slot; {\verb|two timescale issue|}: the bandwidth $\{\omega_s^e(\bar t)\}$ and the beamformers $\{{\bm v}_s(t)\}$ and $\{{\bm g}_{i,s}(t)\}$ should be optimized at two different time scales. $\{\omega_s^e(\bar t)\}$ needs be optimized at the beginning of the time slot. $\{{\bm v}_s(t)\}$ and $\{{\bm g}_{i,s}(t)\}$ should be optimized at the beginning of each minislot.

In the following sections, we discuss \emph{how to address the challenging problem effectively.}

\section{Problem solution with system generated channel coefficients}
In this section, we resort to an SAA technique \cite{kim2015guide} and a distributed optimization method to tackle the above issues.

\subsection{Sample average approximation}
Owing to the ergodicity of channel fading process over the time slot, the objective function can be approximated as
\begin{equation}\label{eq:objfun_approx}
    \frac{1}{T}\sum\limits_{t = 1}^T {{U^e}(t)}  + \frac{1}{T}\sum\limits_{t = 1}^T {\hat \rho {U^u}(t)}  \approx {E_{\hat {\bm h}}}\left[ {{{\hat U}^e} + \hat \rho {{\hat U}^u}} \right],
\end{equation}
where $\hat {\bm h}$ denotes a set of all channel coefficient samples collected at the beginning of the time slot, {$\hat U^e$ and $\hat U^u$ are functions of $\hat {\bm h}$}.

For SAA, its fundamental idea is to approximate the expectation of a random variable by its sample average. The following proposition shows that if the number of samples $M$ is reasonably large, then for all $m \in {\mathcal M} = \{1, \ldots, M\}$, $\{\bar U_m\}$ converges to $\bar U$ uniformly on the feasible region constructed by constraints (\ref{eq:original_problem}b) and (\ref{eq:original_problem}c).

\begin{proposition}\label{pro:pro_saa_convergence}
Let $\Theta$ be a nonempty compact set formed by constraints (\ref{eq:original_problem}b) and (\ref{eq:original_problem}c), $Y({\bm x},{\hat {\bm h}}) = {\hat U}^e + {\hat \rho} {\hat U}^u$ and ${\bm x} = \{ \omega _s^e(\bar t),{\bm v_s}(t),$ ${\bm g_{i,s}}(t)\}$. For any fixed ${\bm x} \in \Theta$, suppose that there exists $\varepsilon  > 0$ such that the family of random variables $\{Y({\bm y},{\hat {\bm h}}) : {\bm y} \in B({\bm x}, \varepsilon)\}$ is uniformly integrable, where $B({\bm x}, \varepsilon)=\{{\bm y}:||{\bm y} - {\bm x}||_2 \le \varepsilon\}$ denotes the closed ball of radius $\varepsilon$ around $\bm x$. Then $\{\bar U_m\}$ converges to $\bar U$ uniformly on $\Theta$ almost surely as $M \to \infty $.
\end{proposition}
\begin{proof}
We omit the proof here as a similar proof can be found in the convergence proof of SAA in \cite{kim2015guide}.
\end{proof}

Based on the conclusion in Proposition 1, given a set of samples of channel coefficients $\{{\bm h}_m\}$ with ${{\bm h}_m} = [{\bm h}_{11,1m};\ldots;{\bm h}_{1J,sm};$ $\ldots; {\bm h}_{(N^e+N^u)J,(|{\mathcal S}^e|+|{\mathcal S}^u|)m}]$ that are assumed to be independent and identically distributed {(i.i.d)}, the original problem (\ref{eq:original_problem}) can be approximated as a single timescale one
\begin{subequations}\label{eq:sample_approx_problem}
\begin{alignat}{2}
& \mathop {{\rm{maximize}}}\limits_{\{ \omega _{sm}^e,{\omega_s^e},{\bm v_{sm}},{\bm g_{i,sm}} \}\hfill} \{\bar U_m\} = \frac{1}{M}\sum\limits_{m = 1}^M {U_m^e }  + \frac{\hat \rho}{M}\sum\limits_{m = 1}^M {U_m^u }  \\
& {\rm s.t:} \text{ }  \omega_{sm}^e = \omega_s^e, \forall s \in {\mathcal S}^e, \forall m \in {\mathcal M}, \\
& \sum\nolimits_{s \in {{\mathcal S}^e}} {{{\bm v}_{j,sm}^{\rm H}}{{\bm v}_{j,sm}}}  + \nonumber \\ 
& \qquad \sum\nolimits_{s \in {\mathcal S}^u} {\sum\nolimits_{i \in {\mathcal I}_s^u} {{{\bm g}_{ij,sm}^{\rm H}}{{\bm g}_{ij,sm}}} }  \le {E_j}, j \in {\mathcal J}, m \in {\mathcal M}, \\
& \sum\nolimits_{s \in {{\mathcal S}^e}} {\omega_{sm}^e}  + W^u(\bm r_m)  \le W, m \in {\mathcal M}, \allowdisplaybreaks[4] \\
& \gamma _{i,sm}^e \ge C_s^{th}, \forall i \in {\mathcal I}_s^e, s \in {\mathcal S}^e, m \in {\mathcal M}, \\
& {\omega_{sm}^e} \ge 0, s \in {\mathcal S}^e, m \in {\mathcal M},
\end{alignat}
\end{subequations}
where $\{\cdot\}_{m}$ denotes a variable corresponding to the $m$-th coefficient sample ${\bm h}_m$. The constraint (\ref{eq:sample_approx_problem}b) is imposed to explicitly describe the two timescale issue of the original problem (\ref{eq:original_problem}). {Note that, the optimization of (\ref{eq:sample_approx_problem}) is not conducted at each minislot $t$; thus, the minislot index $t$ is not involved in (\ref{eq:sample_approx_problem}).}

We consider $\{\omega_{sm}^e\}$ as a family of local variables and $\{\omega_s^e\}$ as a family of global variables in (\ref{eq:sample_approx_problem}). In this case, (\ref{eq:sample_approx_problem}) can be effectively mitigated by distributed optimization routines, such as ADMM \cite{boyd2011distributed}; that is, distributedly optimizing the local variables and then forcing them to the global variables at convergence.


\subsection{Distributed optimization}
According to the fundamental principle of ADMM, the ADMM for (\ref{eq:sample_approx_problem}) can be derived from the following augmented partial Lagrange problem
\begin{subequations}\label{eq:admm_problem}
\begin{alignat}{2}
& \mathop {{\rm{minimize}}}\limits_{\{{ \omega _{sm}^e},{\omega_s^e},{\bm v_{sm}},{\bm g_{i,sm}}\}\hfill}
\sum\limits_{m = 1}^M {\left\{ {- \frac{{U_m^e}}{M} - \frac{{\hat \rho U_m^u}}{M} + } \right.} \nonumber \\
& \left. {\sum\limits_{s \in {{\mathcal S}^e}} {\left[ {\psi _{sm}\left( {\omega _{sm}^e - \omega _s^e} \right) + \frac{\mu }{2}{{\left\| {\omega _{sm}^e - \omega _s^e} \right\|}_2^2}} \right]} } \right\}  \\
& {\rm s.t:} \text{ } \rm{constraints} \text{ } (\ref{eq:sample_approx_problem}c)-(\ref{eq:sample_approx_problem}f) \text{ } \rm{are} \text{ } satisfied.
\end{alignat}
\end{subequations}
where, $\psi_{sm}$ is a Lagrangian multiplier, $\mu$ is a penalty coefficient.

For all channel samples, the distributed framework of mitigating (\ref{eq:admm_problem}) can then be summarized to alternatively calculate equations from (\ref{eq:arg_lagarangian}) to (\ref{eq:psi_update}).
\begin{subequations}\label{eq:arg_lagarangian}
\begin{alignat}{2}
& \left\{ {\begin{array}{*{20}{l}}
{\omega _{sm}^{e(k + 1)},}\\
{{\bm v}_{sm}^{(k + 1)},{\bm g}_{i,sm}^{(k + 1)}}
\end{array}} \right\} = \mathop {{\rm{argmin}}}\limits_{\left\{ {\scriptstyle\omega _{sm}^e,\hfill\atop
\scriptstyle{{\bm v}_{sm}},{{\bm g}_{i,sm}}\hfill} \right\}} \bar {\mathcal L} (\omega _{sm}^e,{{\bm v}_{sm}},{{\bm g}_{i,sm}})  \\
& {\rm s.t:} \text{ } {\rm for} \text{ } {\rm a} \text{ } {\rm sample} \text{ } m, \text{ } (\ref{eq:sample_approx_problem}c)-(\ref{eq:sample_approx_problem}f) \text{ } \rm{are} \text{ } satisfied.
\end{alignat}
\end{subequations}
\begin{equation}\label{eq:omega_update}
\omega _s^{e(k + 1)} = \frac{1}{M}\sum\limits_{m = 1}^M {\left( {\omega _{sm}^{e(k + 1)} + \frac{1}{\mu }\psi _{sm}^{(k)}} \right)}, \text{ } \forall s\in {\mathcal S}^e,
\end{equation}
\begin{equation}\label{eq:psi_update}
    \psi _{sm}^{(k + 1)} = \psi _{sm}^{(k)} + \mu \left( {\omega _{sm}^{e(k + 1)} - \omega _s^{e(k + 1)}} \right), \text{ } \forall s\in {\mathcal S}^e,
\end{equation}
where,
\begin{equation}\label{eq:average_Lagrangian_func}
\begin{array}{l}
    \bar{\mathcal{L}}(\omega_{sm}^e,\bm v_{sm}, \bm g_{i,sm}) =  {- \frac{{U_m^e}}{M} - \frac{{\hat \rho U_m^u}}{M} +} \\
    { \sum\limits_{s \in {{\mathcal S}^e}} {\left[ {\psi _{sm}^{(k)}\left( {\omega _{sm}^e - \omega _s^{e(k)}} \right) + \frac{\mu }{2}{{\left\| {\omega _{sm}^e - \omega _s^{e(k)}} \right\|}_2^2}} \right]} }.
    \end{array}
\end{equation}

In our RAN slicing system, the RAN-C is responsible for executing the distributed framework, and $M$ virtual machines (VMs) are activated to conduct (\ref{eq:arg_lagarangian}) and (\ref{eq:psi_update}). An aggregation VM (AVM) is utilized to aggregate the local variables. Additionally, in this framework, local dual variables $\{\psi_{sm}\}$ are updated to drive local variables $\{\omega_{sm}^e\}$ into consensus, and quadratic items in (\ref{eq:arg_lagarangian}) help pull $\{\omega_{sm}^e\}$ towards their average value.

Unfortunately, the mitigation of (\ref{eq:arg_lagarangian}) on each VM is difficult due to the existence of non-convex constraints. We next attempt to tackle the non-convexity of (\ref{eq:arg_lagarangian}).

\subsection{Semidefinite relaxation scheme}
Let ${\bm V}_{sm} = {\bm v}_{sm} {\bm v}_{sm}^{\rm H} \in {\mathbb R}^{JK \times JK}$ for all $s \in {\mathcal S}^e$, $m \in {\mathcal M}$, and ${\bm G}_{i,sm} = {\bm g}_{i,sm} {\bm g}_{i,sm}^{\rm H} \in {\mathbb R}^{JK \times JK}$ for all $i \in {\mathcal I}_s^u$, $s \in {\mathcal S}^u$, $m \in {\mathcal M}$. Next, if we recall the properties: ${{\bm V_s} = {\bm v_s}\bm v_s^H \Leftrightarrow {\bm V_s} \succeq 0}, \text{ } {{\rm rank}({\bm V_s}) \le 1}$, and ${{\bm G_{i,sm}} = {\bm g_{i,sm}}\bm g_{i,sm}^H \Leftrightarrow {\bm G_{i,sm}} \succeq 0}, \text{ } {{\rm rank}({\bm G_{i,sm}}) \le 1}$, (\ref{eq:arg_lagarangian}) can then be reformulated as
\begin{subequations}\label{eq:arg_lagarangian_reformulated}
\begin{alignat}{2}
& \left\{ {\begin{array}{*{20}{l}}
{\omega _{sm}^{e(k + 1)},}\\
{\bm V_{sm}^{(k + 1)},\bm G_{i,sm}^{(k + 1)}}
\end{array}} \right\} = \mathop {{\rm{argmin}}}\limits_{\left\{ {\scriptstyle\omega _{sm}^e,\hfill\atop
\scriptstyle{\bm V_{sm}},{\bm G_{i,sm}}\hfill} \right\}} \bar {\mathcal L} (\omega _{sm}^e,{\bm V_{sm}},{\bm G_{i,sm}})   \\
& {\rm s.t:} \text{ } \omega _{sm}^e{\log _2}( {1 + \frac{{{\rm tr}({{\bm H}_{i,sm}}{\bm V_{sm}})}}{{\sigma _{i,s}^2}}} ) \ge C_s^{th},\forall s \in {\mathcal S}^e,i \in {\mathcal I}_s^e, \allowdisplaybreaks[4] \\
& \sum\limits_{s \in {{\mathcal S}^e}} {{\rm tr}({{\bm Z}_j{\bm V}_{sm}})}  + \sum\limits_{s \in {{\mathcal S}^u}} {\sum\limits_{i \in {\mathcal I}_s^u} {{\rm tr}({{\bm Z}_j}{{\bm G}_{i,sm}})} }  \le {E_j},\forall j\in {\mathcal J}, \\
& {\bm V_{sm}} \succeq 0, \forall s \in {\mathcal S}^e, \\
& {{\bm G}_{i,sm}} \succeq 0, \forall i \in {\mathcal I}_{s}^u, s \in {\mathcal S}^u,  \\
& {{\rm rank}({\bm V_{sm}}) \le 1}, \forall s \in {\mathcal S}^e,  \\
& {{\rm rank}({{\bm G}_{i,sm}}) \le 1}, \forall i \in {\mathcal I}_{s}^u, s \in {\mathcal S}^u, \\
& {\rm constraints} \text{ } (\ref{eq:sample_approx_problem}\rm{d}),(\ref{eq:sample_approx_problem}\rm{f}) \text{ } {\rm are} \text{ } {\rm satisfied},
\end{alignat}
\end{subequations}
where ${\bm H}_{i,sm} = {\bm h}_{i,sm} {\bm h}_{i,sm}^{\rm H} \in {\mathbb R}^{JK \times JK}$, ${\bm h}_{i,sm} = [{\bm h}_{i1,sm};\ldots;$ ${\bm h}_{iJ,sm}] \in {\mathbb C}^{JK \times 1}$, ${\bm Z}_j \in {\mathbb R}^{JK \times JK}$ is a square matrix with $J \times J$ blocks, and each block in ${\bm Z}_j$ is a $K \times K$ matrix. Besides, in ${\bm Z}_j$, the block in the $j$-th row and $j$-th column is a $K \times K$ identity matrix, and all other blocks are zero matrices.

As power matrices ${\bm V}_{sm}$ ($s\in {\mathcal S}^e$, $m\in {\mathcal M}$) and ${\bm G}_{i,sm}$ ($i \in {\mathcal I}_s^u$, $s \in {\mathcal S}^u$, $m \in {\mathcal M}$) are positive semidefinite, we then resort to the SDR scheme to handle the low-rank non-convex constraints (\ref{eq:arg_lagarangian_reformulated}f) and (\ref{eq:arg_lagarangian_reformulated}g). That is, directly drop the constraints (\ref{eq:arg_lagarangian_reformulated}f) and (\ref{eq:arg_lagarangian_reformulated}g). However, owing to the relaxation, power matrices ${\bm V}_{sm}$ and ${\bm G}_{i,sm}$ obtained by mitigating the problem (\ref{eq:arg_lagarangian_reformulated}) without low-rank constraints will not satisfy the low-rank constraint in general. This is due to the fact that the (convex) feasible set of the relaxed (\ref{eq:arg_lagarangian_reformulated}) is a superset of the (non-convex) feasible set of (\ref{eq:arg_lagarangian_reformulated}).
If they satisfy, then the relaxation is tight; if not, then some manipulation, e.g., \emph{a randomization/scale method} \cite{ma2010semidefinite}, should be performed on them to obtain their approximate solutions.

Although non-convex constraints are removed, constraints related to ${W^u(\bm r_m)}$ are complicated, which hinders the optimization of the relaxed (\ref{eq:arg_lagarangian_reformulated}). Therefore, we next discuss \emph{how to equivalently transform the complicated constraints via a variable slack scheme}.

\subsection{Variable slack scheme}
From (\ref{eq:URLLC_bandwidth}), we observe that $W^u({\bm r}_m)$ is a quadratic function with respect to (w.r.t) $\bm r_m$. Therefore, via introducing a family of slack variables $\bm f_m = \{f_{i,sm}^u\}$, $i \in {\mathcal I}_s^u$, $s \in {\mathcal S}^u$, $m\in {\mathcal M}$, the following lemma shows the equivalent expressions of (\ref{eq:sample_approx_problem}d).
\begin{lemma}\label{lem:urllc_bandwidth_transformation}
    Given the family of slack variables $\bm f_m = \{f_{i,sm}^u\}$, (\ref{eq:sample_approx_problem}d) is equivalent to the following inequalities,
    \begin{equation}\label{eq:slack_equation1}
   {\sum\limits_{s \in {{\mathcal S}^e}} {\omega _{sm}^e}  + A(\bm f_m) + c(\varsigma, \alpha)\sqrt{B(\bm f_m)} \le W},
\end{equation}
and
\begin{equation}\label{eq:slack_equation2}
    f_{i,sm}^u \ge r_{i,sm}^u,
\end{equation}
    for all $i \in {\mathcal I}_s^u$, $s \in {\mathcal S}^u$, $m \in {\mathcal M}$.
\end{lemma}
\begin{proof}
Please refer to Appendix C.
\end{proof}

Besides, we can know that the objective function (\ref{eq:arg_lagarangian_reformulated}a) is convex. This is because it is linear w.r.t variables ${\bm V}_{sm}$ and ${\bm G}_{i,sm}$ with an addition of affine terms and nonnegative quadratic terms w.r.t $\omega_{sm}^e$.
(\ref{eq:arg_lagarangian_reformulated}c) is an affine constraint. Other constraints are non-linear. Based on the above equivalent transformation, we show that (\ref{eq:arg_lagarangian_reformulated}) can be further transformed into a standard convex problem in the following lemma.
\begin{lemma}\label{lem:lemma_SDP_transformation}
    By introducing a family of slack variables, the problem (\ref{eq:arg_lagarangian_reformulated}) without low-rank constraints can be equivalently transformed into the following SDP problem.
    \begin{subequations}\label{eq:standard_transformed_bandwidth_n_beamforming}
\begin{alignat}{2}
& \left\{ {\begin{array}{*{20}{l}}
{\omega _{sm}^{e(k + 1)},\bm V_{sm}^{(k + 1)}}\\
{\bm G_{i,sm}^{(k + 1)},\ldots,\tau_{i,sm}^{u(k+1)}}
\end{array}} \right\} = \mathop {{\rm{argmin}}}\limits_{\left\{ {\scriptstyle\omega _{sm}^e,{\bm V_{sm}},\hfill\atop
\scriptstyle{\bm G_{i,sm}},\ldots,{\tau_{i,sm}^u}\hfill} \right\}} \bar {\mathcal L} (\ldots)  \\
& {\rm s.t:} \text{ } \rm {affine \text{ } constraints \text{ }  (\ref{eq:arg_lagarangian_reformulated}c),(36),(40),(44), (46) }, \\
& \rm {satisfy \text{ } quadratic \text{ } cone \text{ } constraints \text{ }  (38), (39)}, \\
& \rm {satisfy \text{ } exponential \text{ } cones \text{ } (37), (41), (43), (47), (48)}, \\
& \rm {constraints \text{ }  (\ref{eq:arg_lagarangian_reformulated}e),(\ref{eq:arg_lagarangian_reformulated}f) \text{ } are \text{ } satisfied.}
\end{alignat}
\end{subequations}
\end{lemma}
\begin{proof}
Please refer to Appendix D.
\end{proof}

Then, some standard optimization tools such as CVX \cite{CVX} and MOSEK \cite{MOSEK} can be used to mitigate (\ref{eq:standard_transformed_bandwidth_n_beamforming}) effectively.
We can summarize the steps of mitigating (\ref{eq:sample_approx_problem}) in Algorithm \ref{alg1}.

\begin{algorithm}
\caption{Distributed bandwidth optimization algorithm, DBO}
\label{alg1}
\begin{algorithmic}[1]
\STATE \textbf{Input:} randomly initialize ${\omega_{s}^{e(0)}}$, $\psi _{sm}^{(0)}$, $\bm V_{sm}^{(0)}$, for all $i \in {\mathcal I}_s^e$, $s \in {\mathcal S}^e$ and $\bm G_{i,sm}^{(0)}$, for all $i \in {\mathcal I}_s^u$, $s \in {\mathcal S}^u$, ${\bm H}_{i,sm}$, for all $m \in {\mathcal M}$, let $k_{\rm max} =250$.
\STATE \textbf{Output:} $\{\omega_{s}^e\}$
\FOR{$k = 1 : k_{\max}$}
\FOR{each VM $m \in {\mathcal M}$ in parallel}
\STATE VM $m$ solves the problem (\ref{eq:standard_transformed_bandwidth_n_beamforming}) to obtain $\omega_{sm}^{e(k+1)}$ and sends it to the AVM.
\ENDFOR
\STATE After collecting all $\{\omega_{sm}^{e(k+1)}\}$, the AVM aggregates $\omega _s^{e(k + 1)}$ using (\ref{eq:omega_update}) and broadcasts the updated $\omega _s^{e(k + 1)}$ to each VM.
\FOR{each VM $m \in {\mathcal M}$ in parallel}
\STATE VM $m$ computes $\psi _{sm}^{(k + 1)}$ using (\ref{eq:psi_update}) and sends $\psi _{sm}^{(k + 1)}$ to the AVM.
\ENDFOR
\IF{convergence or reach the maximum iteration times $k_{\rm max}$}
\STATE {Break.}
\ENDIF
\ENDFOR
\end{algorithmic}
\end{algorithm}

\subsection{Performance analysis}
In this subsection, we analyze the performance of DBO. We first present a lemma about the optimality of solving (\ref{eq:arg_lagarangian_reformulated}) and then state the computational complexity and the convergence of DBO.

If we denote $\bm G_{i,sm}^{\star}$ and $\bm V_{sm}^{\star}$ as solutions to (\ref{eq:standard_transformed_bandwidth_n_beamforming}), then the following lemma shows the tightness of exploring the SDR scheme on (\ref{eq:arg_lagarangian_reformulated}).
\begin{lemma}\label{lem:SDR}
    {For all $i \in {\mathcal I}_s^u$, $s \in {\mathcal S}^u$, $m \in {\mathcal M}$, the SDR for both ${\bm V}_{sm}$ and ${\bm G}_{i,sm}$} in problem (\ref{eq:arg_lagarangian_reformulated}) is tight, that is,
    \begin{equation}\label{eq:rank_Vs_Gis}
        \begin{array}{l}
        {\rm{rank}}(\bm V_{sm}^ \star ) \le 1,\forall i \in {\mathcal I}_s^e,s \in {{\mathcal S}^e}, \\
        {\rm{rank}}(\bm G_{i,sm}^ \star ) \le 1,\forall i \in {\mathcal I}_s^u,s \in {{\mathcal S}^u}.
        \end{array}
    \end{equation}

    Moreover, $\bm G_{i,sm}^{\star}$ and $\bm V_{sm}^{\star}$ are optimal solutions to (\ref{eq:arg_lagarangian_reformulated}).
\end{lemma}
\begin{proof}
Please refer to Appendix E.
\end{proof}

The computational complexity of DBO is dominated by that of solving the SDP problem. The SDP problem has $\left(|{\mathcal S}^e|+I^u\right)$ matrices of size of $JK \times JK$ and $\left(3|{\mathcal S}^e|+11I^u\right)$ one-dimensional variables. An interior-point method can then be exploited to efficiently mitigate the SDP problem at the worst-case computational complexity of $O(\left(|{\mathcal S}^e|+I^u\right)J^2K^2 +3|{\mathcal S}^e|+11I^u)^{3.5}$ \cite{ye1997interior}. Nevertheless, the actual complexity will usually be much smaller than the worst case.

The following lemma presents the convergence of the algorithm.
\begin{lemma}\label{lem:SDR}
    Let $(\omega _{sm}^{e\star},{\bm v}_{sm}^{\star},{\bm g}_{i,sm}^{\star})$ denote the optimal solutions, under the ADMM-based distributed algorithm, $\forall k \in \mathbb{Z}^+$, $m \in {\mathcal M}$, we have that ${\bar {\mathcal L}(\omega _{sm}^{e(k)},{\bm v}_{sm}^{(k)},{\bm g}_{i,sm}^{(k)})} $ is bounded and
    \begin{equation}\label{eq:Lagrangian_drift}
         \bar {\mathcal L} (\omega _{sm}^{e\star},{\bm v}_{sm}^{\star},{\bm g}_{i,sm}^{\star}) = \mathop {\lim }\limits_{k \to \infty } \bar {\mathcal L} (\omega _{sm}^{e(k)},{\bm v}_{sm}^{(k)},{\bm g}_{i,sm}^{(k)}).
    \end{equation}
\end{lemma}
\begin{proof}
Please refer to Appendix F.
\end{proof}

Further, simulation results show that Algorithm 1 can quickly converge.

\section{Optimization of beamforming with imperfect channel gain}
With the system generated channel coefficient samples, Section V obtains the approximate solution $\{\omega_s^e\}$ to (\ref{eq:original_problem}). In this section, we continue to optimize minislot variables $\{{\bm V}_s(t), {\bm G}_{i,s}(t)\}$ according to sensed imperfect channel gains $\{{\bm H}_{i,s}(t)\}$, $i \in {\mathcal I}^e \cup{{\mathcal I}^u}$, $s \in {\mathcal S}^e \cup{{\mathcal S}^u}$, at the beginning of each minislot $t$.

Given $\{\omega_s^e\}$ and system sensed imperfect channel gains $\{{\bm H}_{i,s}(t)\}$, as the maximization of $U^e(t) + \hat \rho U^u(t)$ at each minislot will lead to the maximization of the time average utility over the whole time slot, the original problem (\ref{eq:original_problem}) can be reduced to the following beamforming optimization problem at each minislot $t$
\begin{subequations}\label{eq:miniSlot_beamforming_problem}
\begin{alignat}{2}
& \mathop {{\rm{maximize}}}\limits_{\{ {\bm V_s}(t),{\bm G_{i,s}}(t)\}}  {{ U}^e(t)}  + {\hat \rho} {{U}^u(t)} \\
& {\rm s.t:} \text{ } \omega _{s}^e{\log _2}( {1 + \frac{{{\rm tr}({{\bm H}_{i,s}(t)}{\bm V_{s}(t)})}}{{\sigma _{i,s}^2}}} ) \ge C_s^{th},\forall s \in {\mathcal S}^e,i \in {\mathcal I}_s^e,  \allowdisplaybreaks[4] \\
& \sum\limits_{s \in {{\mathcal S}^e}} {{\rm tr}({{\bm Z}_j{\bm V}_{s}(t)})}  + \sum\limits_{s \in {{\mathcal S}^u}} {\sum\limits_{i \in {\mathcal I}_s^u} {{\rm tr}({{\bm Z}_j}{{\bm G}_{i,s}(t)})} }  \le {E_j},\forall j\in {\mathcal J}, \\
& {\bm V_s(t)} \succeq 0, \forall s \in {\mathcal S}^e, \\
& {{\bm G}_{i,s}(t)} \succeq 0, \forall i \in {\mathcal I}_{s}^u, s \in {\mathcal S}^u,  \\
& {{\rm rank}({\bm V_{s}(t)}) \le 1}, \forall s \in {\mathcal S}^e,  \\
& {{\rm rank}({{\bm G}_{i,s}(t)}) \le 1}, \forall i \in {\mathcal I}_{s}^u, s \in {\mathcal S}^u, \\
& \rm {constraint \text{ }  (\ref{eq:total_bandwidth}) \text{ } is \text{ } satisfied.}
\end{alignat}
\end{subequations}

By leveraging the presented SDR scheme and variable slack scheme in Section V, (\ref{eq:miniSlot_beamforming_problem}) can be equivalently transformed into a standard SDP problem that is able to be effectively mitigated by CVX or MOSEK.

Recall that the SDR for both $\bm V_{s}(t)$ and $\bm G_{i,s}(t)$ is tight, we therefore can perform the eigenvalue decomposition on $\bm V_{s}(t)$ and $\bm G_{i,s}(t)$ to obtain the optimal beamforming vectors $\bm v_{s}(t)$ and $\bm g_{i,s}(t)$, respectively.

Then, the bandwidth and beamforming optimization algorithm designed for the RAN slicing system can be summarized as follows.
\begin{algorithm}
\caption{Bandwidth and beamforming optimization algorithm based on ADMM, B$^2$O-ADMM }
\label{alg_algo_bandwidth_beamforming}
\begin{algorithmic}[1]
\STATE \textbf{Input:} $\{{\bm H}_{i,s}(t)\}$, for all $i \in {\mathcal I}^e \cup{{\mathcal I}^u}$, $s \in {\mathcal S}^e \cup{{\mathcal S}^u}$
\STATE \textbf{Output:} $\{\omega_s^e\}$, $\{\bm v_{s}(t)\}$, and $\{\bm g_{i,s}(t)\}$
\STATE Call Algorithm \ref{alg1} to generate $\{\omega_s^e\}$, for all $s\in {\mathcal S}^e$.
\FOR{$t = 1 : T$}
\STATE Given $\{\omega_s^e\}$, the RAN-C mitigates (\ref{eq:miniSlot_beamforming_problem}) to obtain beamformers $\{\bm v_{s}(t)\}$ for all $s\in {\mathcal S}^e$ and $\{{\bm g}_{i,s}(t)\}$ for all $i \in {\mathcal I}_s^u$, $s \in {\mathcal S}^u$.
\ENDFOR
\end{algorithmic}
\end{algorithm}

\section{Simulation results}
\subsection{Comparison algorithms and simulation setup}
We {implement} the following algorithms {on Python and compare them} to evaluate the performance of bandwidth allocation and beamforming algorithms in the RAN slicing system: i) the proposed B$^2$O-ADMM algorithm; ii) the IRHS algorithm in \cite{tang2019service}, which enforces all slices requests and optimizes the same objective function as B$^2$O-ADMM; iii) the proposed resource allocation algorithm without ADMM, NoADMM. Specifically, NoADMM algorithm generates the bandwidth allocated to eMBB slices based on imperfect channel gains sensed at the beginning of the $1^{\rm {st}}$ minislot.

In the simulation, {we consider the CoMP-enabled RAN slicing system with three RRHs located on a circle with a radius of 0.5 km. The distance between every two RRHs is equal.}
eMBB and URLLC UEs are randomly and uniformly distributed in the circle. The transmit antenna gain at each RRH is set to be $5$dB, and a log-normal shadowing path loss model is utilized to simulate the path loss between a RRH and a UE. Particular, a downlink path loss is calculated by $H({\rm dB}) = 128.1 + 37.6\log_{10}d$, where $d$ (in km) represents the distance between a UE and a RRH. The log-normal shadowing standard deviation is set to be 10 dB. Besides, we let the maximum transmit power $E_1 = E_2 =E_3 = 1$ W, $\sigma_{i,s}^2 = -110$ dBm for all $i \in {\mathcal I}^e\cup{{\mathcal I}^u}$, $s \in {\mathcal S}^e\cup{{\mathcal S}^u}$, $L_{i,s}^u = 160$ bits, $\lambda_{s} = \lambda = 0.1$ packet per unit time for all $i \in {\mathcal I}_s^u$, $s\in {\mathcal S}^u$, $K=2$, $\eta=1000$, $\hat \rho=500$, $M = 100$, $T=60$ minislots, $\kappa=5.12 \times 10^{-4}$ channel uses per unit time per unit bandwidth, $W=4$ MHz, $\phi=1.5$. Other slice configuration parameters are listed as below: $\varsigma=2\times10^{-5}$, $\alpha = 10^{-5}$, $\beta = 2 \times 10^{-8}$, $|{\mathcal S}^e|=3$, $|{\mathcal S}^u| =2$, $\{I_s^e\} = \{4, 6, 8\}$ UEs, $C_s^{th} = \{6, 4, 2\}$ Mb/s, $\{I_s^u\} = \{3, 5\}$ UEs, and $\{D_s\} = \{1, 2\}$ milliseconds \cite{tang2019service}.

\subsection{Performance evaluation}
The following performance indicators are adopted to evaluate the comparison algorithms: i) system bandwidth $W_u$ (in MHz) allocated to URLLC slices; ii) total transmit power $E^u = \sum\nolimits_{t=1}^{\rm T} \sum\nolimits_{s \in {{\mathcal S}^u}} {\sum\nolimits_{i \in {\mathcal I}_s^u} {{\rm tr}({{\bm G}_{i,s}(t)})} } $ (in W) configured for URLLC slices; iii) long-term total slice utility $\bar U$ that is the objective function of (\ref{eq:original_problem}).

We first evaluate the convergence of the proposed B$^2$O-ADMM algorithm. As shown in Algorithm \ref{alg_algo_bandwidth_beamforming}, the convergence of B$^2$O-ADMM is solely determined by that of the ADMM-based DBO algorithm. We thus evaluate the convergence of DBO instead. Denote the loss of accuracy of the global consensus variable by $\Delta_{\omega} = {\rm{ }}\sum\nolimits_{s \in {{\mathcal S}^e}} {\left| {\omega _s^{e(k + 1)} - \omega _s^{e(k)}} \right|} $. According to the principle of ADMM, if the loss value $\Delta_{\omega}$ approaches zero after a limited number of iterations, then the ADMM-based DBO algorithm is convergent; otherwise, the algorithm is divergent. Fig. \ref{fig:fig_convergence} depicts the convergence curve of the proposed B$^2$O-ADMM algorithm. It shows that B$^2$O-ADMM can converge after several iterations.
\begin{figure}[!t]
\centering
\includegraphics[width=2.6in]{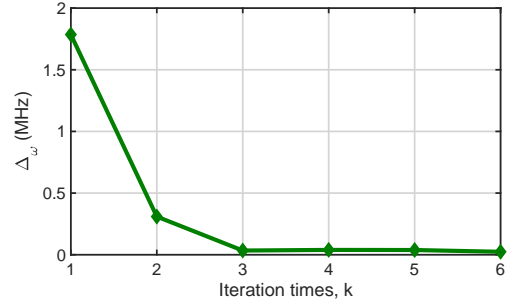}
\caption{A convergence curve of the B$^2$O-ADMM algorithm.}
\label{fig:fig_convergence}
\end{figure}

To understand the impact of packet arrival rate on the performance of comparison algorithms, we plot the relationship between the bandwidth allocated to URLLC slices and the packet arrival rate $\lambda$ with $\lambda = \{0.1,0.2,\ldots,1.0,1.1\}$ packets per unit time in Fig. \ref{fig:fig_bandwidth_vs_lambda}. Besides, the trends of the transmit power configured for URLLC slices and the achieved total slice utility over both $\lambda$ and energy efficiency coefficient $\hat \rho$ are illustrated in Fig. \ref{fig:fig_power_n_utility_vs_lambda}.

From these two figures, we have the following observations:
i) for B$^2$O-ADMM and NoADMM algorithms, the system bandwidth allocated to URLLC slices monotonously increases with an increasing arrival rate when the RAN slicing system provides URLLC and eMBB multiplexing services. When the system terminates eMBB services, B$^2$O-ADMM and NoADMM algorithms recommend to allocate the total network bandwidth to URLLC slices to guarantee more reliable URLLC transmission;
ii) for the IRHS algorithm, as it does not design some strategies to reduce the packet blocking probability of URLLC packets, it suggests keeping the amount of system bandwidth allocated to URLLC slices at a low constant value. Meanwhile, the IRHS algorithm will not allocate the total network bandwidth to URLLC slices even though in the case of end of eMBB services;
iii) compared with B$^2$O-ADMM and NoADMM, although IRHS needs less network bandwidth to ensure an ultra-low URLLC decoding error probability, it desires greater transmit power $E^u$ to satisfy QoS requirements of URLLC UEs and then consumes more energy;
iv) the B$^2$O-ADMM algorithm obtains the greatest long-term total slice utility. Besides, it can be utilized to configure a slicing system, which can support URLLC transmission with higher arrival rates without significantly decreasing the achieved total slice utility $\bar U$.
\begin{figure}[!t]
\centering
\includegraphics[width=2.9in]{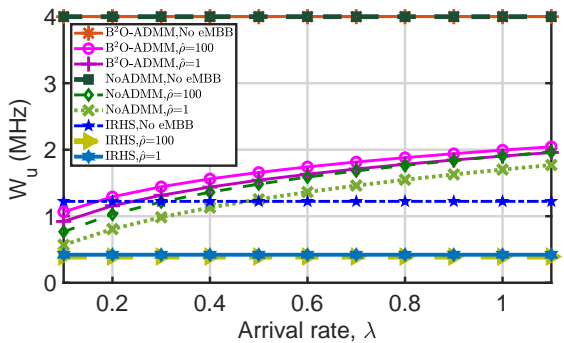}
\caption{Trend of the system bandwidth allocated to URLLC slices under different arrival rates, $\lambda$.}
\label{fig:fig_bandwidth_vs_lambda}
\end{figure}

\begin{figure}[!t]
\centering
\subfigure[Total transmit power configured for URLLC slices]{\includegraphics[width=2.9in]{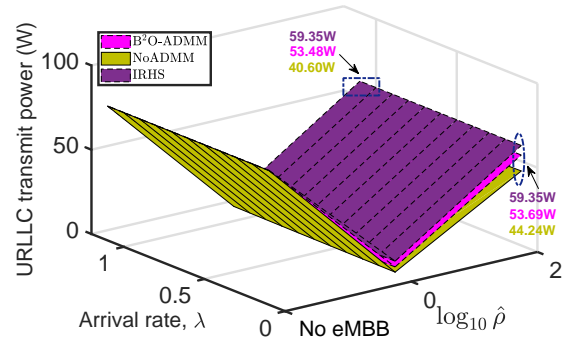}%
\label{fig:fig_total_utility_vs_lamdba_IoT}}
\hspace{0.05\linewidth}
\subfigure[Achieved long-term total slice utility]{\includegraphics[width=2.9in]{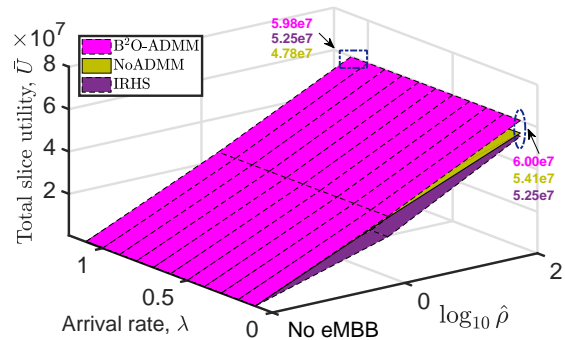}%
\label{fig_vs_Ds}}
\caption{Trends of transmit power configured for URLLC slices and achieved long-term total slice utility vs. $\hat \rho$ and $\lambda$.}
\label{fig:fig_power_n_utility_vs_lambda}
\end{figure}

We next evaluate the impact of the slice priority coefficient $\hat \rho$ on the performance of the comparison algorithms with $\hat \rho \in \{1, 50, 100, \ldots, 500\}$. In Fig. \ref{fig:fig_total_utility_vs_rho}, we plot the trend of achieved total slice utility $\bar U$ over different slice priority coefficients $\hat \rho$ with $\lambda = 0.1$ and $\eta = 1000$. Besides, Fig. \ref{fig:fig_bandwidth_power_vs_rho} shows the system bandwidth $W_u$ allocated to URLLC slices and configured URLLC transmit power $E^u$ under different $\hat \rho$.
\begin{figure}[!t]
\centering
\includegraphics[width=2.9in]{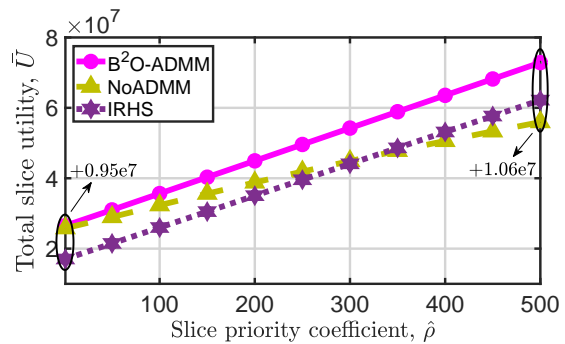}
\caption{Trend of achieved long-term total slice utility vs. $\hat \rho$.}
\label{fig:fig_total_utility_vs_rho}
\end{figure}
\begin{figure}[!t]
\centering
\includegraphics[width=2.9in]{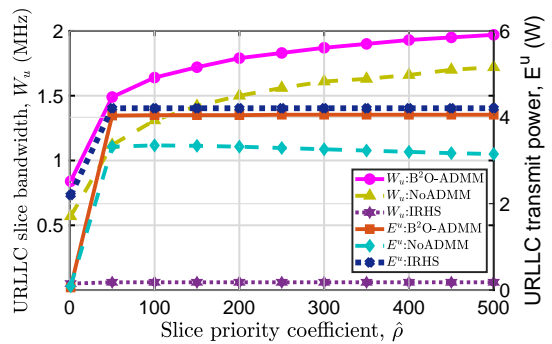}
\caption{The configured $W^u$ and $E^u$ of all comparison algorithms under different $\hat \rho$}
\label{fig:fig_bandwidth_power_vs_rho}
\end{figure}

We can obtain the following observations from these figures:
i) the proposed B$^2$O-ADMM can obtain the greatest $\bar U$. The achieved $\bar U$ monotonously increases with $\hat \rho$ for all comparison algorithms;
ii) the total slice utility gain of B$^2$O-ADMM over IRHS does not significantly change with $\hat \rho$. The utility gain of B$^2$O-ADMM over NoADMM increases with an increasing $\hat \rho$. Compared with NoADMM, a great $\hat \rho$ results in large bandwidth $W_u$ and then a great URLLC slice utility {can be obtained by B$^2$O-ADMM};
iii) for the B$^2$O-ADMM and NoADMM algorithms, they suggest configuring great system bandwidth for URLLC slices when URLLC slices have a great slice priority coefficient. Compared with other algorithms, B$^2$O-ADMM requires the largest system bandwidth to accommodate the QoS requirements of URLLC UEs. For IRHS, the obtained $W_u$ is robust to $\hat \rho$;
iv) great transmit power should be configured if the URLLC slice priority coefficient is great. {However, when $\hat \rho > 50$, the obtained $E^u$ of B$^2$O-ADMM and IRHS are robust to $\hat \rho$. For NoADMM, its obtained $E^u$ slightly decreases with an increasing $\hat \rho$ as the corresponding system bandwidth allocated to URLLC slices is widened.}

At last, to understand the effect of energy efficiency coefficient $\eta$, we plot the relationship between $\bar U$ and $\eta$ in Fig. \ref{fig:fig_total_utility_vs_eta} and plot trends of $W^u$ and $E^u$ under different $\eta$ in Fig. \ref{fig:fig_urllc_bandwidth_n_transmit_power} with $\lambda = 0.1$ and $\hat \rho = 500$.
\begin{figure}[!t]
\centering
\includegraphics[width=2.9in]{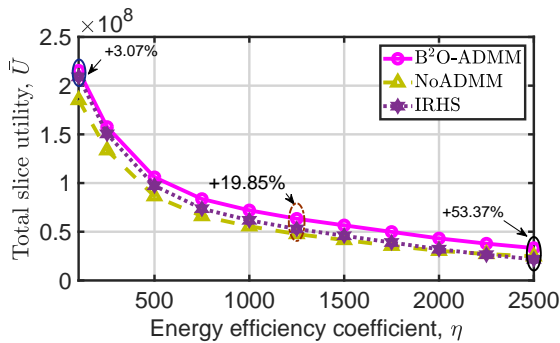}
\caption{Trend of the long-term total slice utility vs. $\eta$.}
\label{fig:fig_total_utility_vs_eta}
\end{figure}
\begin{figure}[!t]
\centering
\includegraphics[width=2.9in]{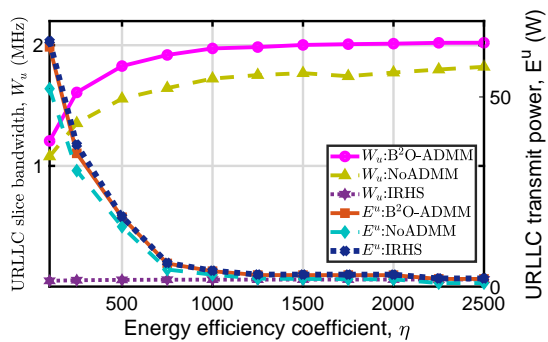}
\caption{Trends of URLLC slice bandwidth $W_u$ and transmit power $E^u$ configured for URLLC slices vs. $\eta$.}
\label{fig:fig_urllc_bandwidth_n_transmit_power}
\end{figure}

We can observe the following conclusions from Figs. \ref{fig:fig_total_utility_vs_eta} and \ref{fig:fig_urllc_bandwidth_n_transmit_power}:
i) the proposed B$^2$O-ADMM algorithm achieves the greatest $\bar U$, and the obtained total slice utilities of all comparison algorithms monotonously decrease as $\eta$ increases. A great $\eta$ indicates that the power consumption dominates the total slice utility; thus, all comparison algorithms reduce the power consumption and the corresponding SNR received at each end UE;
ii) for B$^2$O-ADMM and NoADMM, their obtained $W_u$ increase with an increasing $\eta$. A large $\eta$ leads to the system configuration of small transmit power; thus, the system bandwidth allocated to URLLC slices should be widened to satisfy QoS requirements of URLLC UEs. As QoS requirements of URLLC UEs are easy to be satisfied for the IRHS algorithm, a decreasing $E^u$ does not result in a significantly increasing $W_u$;
iii) compared with NoADMM, B$^2$O-ADMM suggests apportioning greater network bandwidth to URLLC slices as B$^2$O-ADMM achieves the global consensus bandwidth to ensure an ultra-low packet blocking probability over the whole time slot.

\section{Conclusion}
In this paper, we considered a CoMP-enabled RAN slicing system simultaneously supporting URLLC and eMBB traffic transmission. In the presence of eMBB traffic, we orchestrated the shared network resources of the system to guarantee a more reliable bursty URLLC service provision from the perspectives of lowering both URLLC packet blocking probability and codeword decoding error probability. We formulated the problem of RAN slicing for bursty URLLC and eMBB service multiplexing as a resource optimization problem and developed a joint bandwidth and CoMP beamforming optimization algorithm to maximize the long-term total slice utility.

There are some interesting directions to explore in the future, such as validating the performance of the proposed algorithm through experiments and conducting the concrete prototype implementation of the RAN slicing system.



\appendix

\subsection{Proof of \textbf{Lemma} 1}
Given an $M/M/W^u$ queueing system with a URLLC packet arrival rate $\bm \lambda$ and a packet transmit speed of $\{\kappa/r_{i,s}^u(t)\}$. The QoS goal of configuring URLLC slices is that the queueing probability $P_Q^{M/M/W^u}$ in the $M/M/W^u$ system is lower than a given value $\varsigma$ and the packet blocking probability $p_b({\bm \omega}^u, {\bm \lambda}, {\bm D}, W^u({\bm r}(t)))$ (or $p_b$ for notation lightening) is of order $\alpha$. To achieve this goal, we exploit the square-root staffing rule \cite{harchol2013performance} to derive the needed network bandwidth.

If we let $A({\bm r}(t)) = \sum\limits_{s \in {\mathcal S}^u } \sum\limits_{i \in {\mathcal I}_s^u} {\lambda _{s}}{\frac{{{{r_{i,s}^u(t)}}}}{\kappa}} $, $B({\bm r}(t)) = \sum\limits_{s \in {\mathcal S}^u} \sum\limits_{i \in {\mathcal I}_s^u}{\lambda _{s}} \times$ ${\frac{{{r_{i,s}^u(t)}^2}}{{{\kappa ^2}{D_s}}}}$, then the minimum network bandwidth needed to satisfy the QoS goal can be approximately expressed as \cite{harchol2013performance}
\begin{equation}\label{eq:approx_bandwidth_appendix}
    W^u({\bm r}(t))  \approx  A({\bm r}(t)) + c(\varsigma, \alpha)\sqrt{B({\bm r}(t))}
\end{equation}

For the $M/M/W^u$ queueing system, the expression of $P_Q^{M/M/W^u}$ w.r.t $p_b$ can be written as \cite{harchol2013performance}
\begin{equation}\label{eq:queueing_block_probability}
    P_Q^{M/M/W^u} = \frac{\left (A({\bm r}(t))+c(\varsigma, \alpha)\sqrt{B({\bm r}(t))}\right )p_b}{c(\varsigma, \alpha)\sqrt{B({\bm r}(t))} + A({\bm r}(t))p_b}
\end{equation}

Since $P_Q^{M/M/W^u} \le \varsigma$ and $\varsigma > \alpha$, (\ref{eq:queueing_block_probability}) can take the following form
\begin{equation}\label{eq:derive_PQ_1}
    c(\varsigma, \alpha) \ge \frac{{A(r(t))}}{{\sqrt {B(r(t))} }}\frac{{{\alpha} - \varsigma {\alpha}}}{{\varsigma  - {\alpha}}}
\end{equation}

According to Cauchy-chwarz inequality, we can scale up $A({\bm r}(t))$ as
\begin{equation}\label{eq:derive_PQ_2}
    A({\bm r}(t)) \le \sqrt {\sum\limits_{s \in {{\mathcal S}^u}} {\sum\limits_{i \in {\mathcal I}_s^u} {\lambda _s^2D_s^2} } } \sqrt {\sum\limits_{s \in {{\mathcal S}^u}} {\sum\limits_{i \in {\mathcal I}_s^u} {\frac{{r_{i,s}^u{{(t)}^2}}}{{{\kappa ^2}D_s^2}}} } }
\end{equation}

For $B({\bm r}(t))$, it can be scaled down as
\begin{equation}\label{eq:derive_PQ_3}
    B(r(t)) \ge \mathop {\min }\limits_{s \in {{\mathcal S}^u}} \{ {\lambda _s}{D_s}\} \sum\limits_{s \in {{\mathcal S}^u}} {\sum\limits_{i \in {\mathcal I}_s^u} {\frac{{r_{i,s}^u{{(t)}^2}}}{{{\kappa ^2}{D_s^2}}}} }
\end{equation}

It can be observed that $W^u({\bm r}(t))$ monotonously increases with $c(\varsigma, \alpha)$. A great $c(\varsigma, \alpha)$ indicates that more bandwidth should be cut and apportioned to URLLC slices, and $W^u({\bm r}(t))$ will obtain the minimum value when (\ref{eq:derive_PQ_1}) is active. Therefore, on the basis of (\ref{eq:derive_PQ_1})-(\ref{eq:derive_PQ_3}), to save network bandwidth while ensuring a low packet blocking probability we have
\begin{equation}\label{eq:derive_PQ_4}
    c(\varsigma, \alpha) = \frac{{{\alpha} - \varsigma {\alpha}}}{{\varsigma  - {\alpha}}}\sqrt {\frac{{\sum\limits_{s \in {{\mathcal S}^u}} {I_s^u\lambda _s^2D_s^2} }}{{\mathop {\min }\limits_{s \in {{\mathcal S}^u}} \{ {\lambda _s}{D_s}\} }}}
\end{equation}

This completes the proof.

\subsection{Proof of \textbf{Lemma} 2}
By calculating the first order derivation of $r_{i,s}^{u}(t)$ over $V(SNR_{i,s}^u(t))$, we observe that $r_{i,s}^{u}(t)$ monotonously increases with $V(SNR_{i,s}^u(t))$. Since $\ln^2 2$ is the maximum value of $V(SNR_{i,s}^u(t))$, if we let $V(SNR_{i,s}^u$ $(t))=\ln^2 2$, then we can obtain the minimum upper bound of $r_{i,s}^{u}(t)$.

Let $\sqrt{r_{i,s}^u(t)} = x$ and $V(SNR_{i,s}^u(t))=\ln^2 2$, then (\ref{eq:URLLC_bit_length}) becomes a quadratic equation w.r.t $x$. Solving it we can achieve the closed-form expression for the minimum upper bound of $r_{i,s}^u(t)$ that is shown in (\ref{eq:URLLC_channel_use}).
This completes the proof.

\subsection{Proof of \textbf{Lemma} 3}
On the one hand, if the constraint (\ref{eq:slack_equation2}) is active, then (\ref{eq:slack_equation1}) and (\ref{eq:slack_equation2}) are equivalent to (\ref{eq:sample_approx_problem}d); on the other hand, if at the optimal solution to (\ref{eq:arg_lagarangian_reformulated}) constrained by (\ref{eq:slack_equation1}) and (\ref{eq:slack_equation2}), there is a sample $m$ (or UE $i \in {\mathcal I}_s^u$, $s \in {\mathcal S}^u$) such that (\ref{eq:slack_equation2}) is non-active, then we can always pull the value of $f_{i,sm}^u$ towards $r_{i,sm}^u$ without violating (\ref{eq:slack_equation1}) and changing the value of the objective function. The constraints (\ref{eq:slack_equation1}) and (\ref{eq:slack_equation2}) are therefore equivalent to (\ref{eq:sample_approx_problem}d). This completes the proof.

\subsection{Proof of \textbf{Lemma} 4}
In this subsection, we exploit the variable slack scheme to transform non-linear constraints into convex cones.

For the constraint (\ref{eq:arg_lagarangian_reformulated}b), we introduce a variable ${{\bar \lambda }_{i,sm}}$ and let
\begin{equation}\label{eq:ln_rate}
    \ln \left( {1 + \frac{{{\rm tr}({\bm H_{i,sm}}{{\bm V}_{sm}})}}{{\sigma _{i,s}^2}}} \right) \ge {\bar \lambda _{i,sm}}\ln 2,\forall s \in {{\mathcal S}^e},i \in {\mathcal I}_s^e
\end{equation}

By introducing the variable $\{{{\theta }_{i,sm}}\}$, we can obtain
\begin{equation}\label{eq:linear_snr_theta}
    {{{\rm tr}({\bm H_{i,sm}}{{\bm V}_{sm}})}}/{{\sigma _{i,s}^2}} \ge {\theta _{i,sm}},\forall s \in {{\mathcal S}^e},i \in {\mathcal I}_s^e
\end{equation}

(\ref{eq:ln_rate}) can then be rewritten as a standard convex expression, i.e.,
\begin{equation}\label{eq:standard_theta_ism}
    (1 + {\theta _{i,sm}},1,{\bar \lambda _{i,sm}}\ln 2) \in {{\mathcal K}_{\exp }},\forall s \in {{\mathcal S}^e},i \in {\mathcal I}_s^e
\end{equation}
where ${{\mathcal K}_{\exp }} = \{ ({x_1},{x_2},{x_3}):{x_1} \ge {x_2}{e^{{x_3}/{x_2}}},{x_2} > 0\} $ is an exponential cone of ${\mathbb R}^3$.

$\omega_{sm}^e$ can then be correlated with ${\bar \lambda _{i,sm}}$ by the following quadratic cone
\begin{equation}\label{eq:standard_omega}
    (\omega _{sm}^e,{\bar \lambda _{i,sm}},\sqrt {2C_s^{th}} ) \in {\mathcal Q}_r^3,\forall s \in {{\mathcal S}^e},i \in {\mathcal I}_s^e
\end{equation}
where ${\mathcal Q}_r^n = \{{\bm x}|2{x_1}{x_2} \ge x_3^2 +  \ldots  + x_n^2,{x_1},{x_2} \ge 0\} $ is a rotated quadratic cone of ${\mathbb R}^n$.

For ${\bm x} \in {\mathbb R}^n$, as $t \ge {\left\| {\bm x} \right\|_2} \Leftrightarrow (t,{\bm x}) \in {{\mathcal Q}^{n + 1}}$, the constraint (\ref{eq:slack_equation1}) is equivalent to the following expression
\begin{equation}\label{eq:standard_f_mean}
    ( {c^{-1}(\varsigma,\alpha )(W - \sum\limits_{s \in {{\mathcal S}^e}} {\omega _{sm}^e}  - A({\bm f_m})),\{ {\frac{\sqrt{\lambda_s}{f_{i,sm}^u}}{{\kappa \sqrt{D_s}}}} \}} ) \in {{\mathcal Q}^{\sum\limits_{s\in {\mathcal S}^u} {I_s^u}  + 1}}
\end{equation}
where ${{\mathcal Q}^n} = \{x| {x_1} \ge \sqrt {x_2^2 +  \ldots  + x_n^2} \} $ is a quadratic cone of ${\mathbb R}^n$.

Next, we let ${e^{d_{i,sm}^u}} \ge \frac{1}{{C(SNR_{i,sm}^u)}},\forall i \in {\mathcal I}_s^u,s \in {{\mathcal S}^u}$ and introduce the variable $\{\tau _{i,sm}^{u}\}$. In this way, we can obtain
\begin{equation}\label{eq:linear_snr_tau_ism}
    {\rm{tr}}({\bm H_{i,sm}}{{\bm G}_{i,sm}})/\phi \sigma _{i,s}^2 \ge \tau _{i,sm}^u,\forall i \in {\mathcal I}_s^u,s \in {{\mathcal S}^u}
\end{equation}
and
\begin{equation}\label{eq:standard_phi_d}
    \left\{ {\begin{array}{*{20}{c}}
{(\varphi _{i,sm}^u,1, - d_{i,sm}^u + \ln \ln 2) \in {{\mathcal K}_{\exp }}}\\
{(1 + \tau _{i,sm}^u,1,\varphi _{i,sm}^u) \in {{\mathcal K}_{\exp }}}
\end{array}} \right.,\forall i \in {\mathcal I}_s^u,s \in {{\mathcal S}^u}
\end{equation}
where (\ref{eq:standard_phi_d}) stems from the fact that $\ln \ln x \ge t,x > 1 \Leftrightarrow \{ (u,1,t) \in {{\mathcal K}_{\exp }},(x,1,u) \in {{\mathcal K}_{\exp }}\} $.

For constraint (\ref{eq:slack_equation2}), we introduce the variable $\{\nu_{i,sm}^{u}\}$ with ${e^{\nu_{i,sm}^u}} \ge \sqrt {1 + \frac{{4L_{i,s}^u{e^{ - d_{i,sm}^u}}}}{Y}} $, where $Y = {({Q^{ - 1}}(\beta){\ln ^{ - 1}}2)^2}$. Then we have
\begin{equation}\label{eq:w_ism}
    2\nu_{i,sm}^u \ge \ln ({e^0} + {e^{ - d_{i,sm}^u + \ln (4L_{i,s}^u/Y)}})
\end{equation}

Since $t \ge \ln (\sum\nolimits_i {{e^{{x_i}}}} ) \Leftrightarrow \{ \sum\nolimits_i {{\mu _i}}  \le 1;({\mu _i},1,{x_i} - t) \in {{\mathcal K}_{\exp }},\forall i\} $, (\ref{eq:w_ism}) can be rewritten as
\begin{equation}\label{eq:standard_mu_nu}
    \left\{ {\begin{array}{*{20}{c}}
{\left( {\mu _{i,sm}^{u;1},1, - 2\nu_{i,sm}^u} \right) \in {{\mathcal K}_{\exp }}}\\
{\left( {\mu _{i,sm}^{u;2},1, - d_{i,sm}^u - 2\nu_{i,sm}^u + \ln (4L_{i,s}^u/Y)} \right) \in {{\mathcal K}_{\exp }}}
\end{array}} \right.
\end{equation}
\begin{equation}\label{eq:standard_linear_mu}
    {\mu _{i,sm}^{u;1} + \mu _{i,sm}^{u;2} \le 1}
\end{equation}

Besides, by letting $f_{i,sm}^u \ge {e^{x_{i,sm}^u}}$, (\ref{eq:slack_equation2}) can be transformed into the following inequality
\begin{equation}\label{eq:x_ism}
    \begin{array}{l}
x_{i,sm}^u \ge \ln \left( {{e^{d_{i,sm}^u + \ln L_{i,s}^u}} + {e^{2d_{i,sm}^u + \ln (Y/2)}} + } \right.\\
\qquad \quad \left. {{e^{2d_{i,sm}^u + \nu_{i,sm}^u + \ln (Y/2)}}} \right),\forall i \in {\mathcal I}_s^u,s \in {{\mathcal S}^u}
\end{array}
\end{equation}

Likewise, (\ref{eq:x_ism}) can take the following forms
\begin{equation}\label{eq:linear_zeta_ism}
    {\zeta _{i,sm}^{u;1} + \zeta _{i,sm}^{u;2} + \zeta _{i,sm}^{u;3} \le 1}
\end{equation}
\begin{equation}\label{eq:standard_zeta_ism}
    \left\{ {\begin{array}{*{20}{c}}
{\left( {\zeta _{i,sm}^{u;1},1,d_{i,sm}^u - x_{i,sm}^u + \ln L_{i,s}^u} \right) \in {{\mathcal K}_{\exp }}}\\
{\left( {\zeta _{i,sm}^{u;2},1,2d_{i,sm}^u - x_{i,sm}^u + \ln (Y/2)} \right) \in {{\mathcal K}_{\exp }}}\\
{\left( {\zeta _{i,sm}^{u;3},1,2d_{i,sm}^u + \nu_{i,sm}^u - x_{i,sm}^u + \ln (Y/2)} \right) \in {{\mathcal K}_{\exp }}}
\end{array}} \right.
\end{equation}

At last, for the inequality $f_{i,sm}^u \ge {e^{x_{i,sm}^u}}$, it can take the following exponential cone expression
\begin{equation}\label{eq:standard_f_ism}
    \left( {f_{i,sm}^u,1,x_{i,sm}^u} \right) \in {{\mathcal K}_{\exp }},\forall i \in {\mathcal I}_s^u,s \in {{\mathcal S}^u}
\end{equation}

Next, we should prove the equivalence of transforming the above non-linear constraints. As a similar proof for the equivalent transformation of constraints via the variable slack scheme can be found in subsection V-D, we omit the proof here for brevity.

At this point, we may say that the problem (\ref{eq:arg_lagarangian_reformulated}) without low-rank constraints can be equivalently transformed into the problem (\ref{eq:standard_transformed_bandwidth_n_beamforming}). Further, it can observe that (\ref{eq:standard_transformed_bandwidth_n_beamforming}) consists of a quadratic objective function, a set of affine constraints, (rotated) quadratic cone constraints and convex cone constraints of positive semidefinite matrices. Next, by referring to the fact that semidefinite optimization is a generalization of conic optimization, which allows for the utilization of matrix variables belonging to the convex cone of positive semidefinite matrices, we can conclude that the problem (\ref{eq:standard_transformed_bandwidth_n_beamforming}) is an SDP problem. This completes the proof.

\subsection{Proof of \textbf{Lemma} 5}
Although some standard optimization tools were leveraged to mitigate (\ref{eq:standard_transformed_bandwidth_n_beamforming}), they could not capture structural features (e.g., the rank) of the optimal solution. We resort to the Lagrange dual method to prove the tightness of SDR for power matrices.

The Lagrange dual problem of (\ref{eq:standard_transformed_bandwidth_n_beamforming}) can be formulated as
\begin{subequations}\label{eq:lagrange_dual}
\begin{alignat}{2}
& \mathop {\max }\limits_{\left\{ {\scriptstyle{{\bar \varphi }_{i,sm}},{{\bar \chi }_{i,sm}},\hfill\atop
\scriptstyle{{\bar \mu }_{jm}},{{\bm {\Phi}} _{sm}},{{\bm X}_{i,sm}}\hfill} \right\}} \mathop {\min }\limits_{\left\{ {\scriptstyle\omega _{sm}^e,{\bm V_{sm}},\hfill\atop
\scriptstyle{\bm G_{i,sm}},\ldots,{\tau_{i,sm}^u}\hfill} \right\} \in {\mathcal F}_m} L( \ldots )\\
& {\rm subject \text{ } to:} \nonumber \\
& {{\bar \varphi }_{i,sm}} \ge 0,{{\bm {\Phi}} _{sm}} \succeq 0,\forall i \in {\mathcal I}_s^e,s \in {{\mathcal S}^e}\\
& {{\bar \chi }_{i,sm}} \ge 0,{{\bm X}_{i,sm}} \succeq 0,\forall i \in {\mathcal I}_s^u,s \in {{\mathcal S}^u}\\
& {{\bar \mu }_{jm}} \ge 0, j \in {\mathcal J}
\end{alignat}
\end{subequations}
where ${\mathcal F}_m$ is the feasible region configured by constraints (\ref{eq:standard_transformed_bandwidth_n_beamforming}c)-(\ref{eq:standard_transformed_bandwidth_n_beamforming}e), (\ref{eq:standard_linear_mu}), and (\ref{eq:linear_zeta_ism}), and the partial Lagrangian function
\begin{equation}\label{eq:lagrangian_func}
    \begin{array}{l}
L( \ldots ) =  - \left( {\frac{1}{M} + {{\bar \varphi }_{i,sm}}} \right)\sum\limits_{s \in {{\mathcal S}^e}} {\sum\limits_{i \in {\mathcal I}_s^e} {\frac{{{\rm tr}({\bm H_{i,sm}}{{\bm V}_{sm}})}}{{\sigma _{i,s}^2}}} }  + \\
\sum\limits_{s \in {{\mathcal S}^e}} {\left[ {\frac{\eta }{M}{\rm tr}({{\bm V}_{sm}}) + \sum\limits_{j \in {\mathcal J}} {{{\bar \mu }_{jm}}{\rm{tr}}({\bm Z_j}{{\bm V}_{sm}})}  - {\rm tr}({\bm \Phi _{sm}^{\rm T}}{{\bm V}_{sm}})} \right]}  - \\
\sum\limits_{s \in {{\mathcal S}^u}} {\sum\limits_{i \in {\mathcal I}_s^u} {\left[ {\left( {\frac{{\hat \rho }}{M} + {{\bar \chi }_{i,sm}}} \right)\frac{{{\rm tr}({\bm H_{i,sm}}{{\bm G}_{i,sm}})}}{{\phi \sigma _{i,s}^2}} + {\rm tr}({{\bm X}_{i,sm}^{\rm T}}{{\bm G}_{i,sm}})} \right]} }  + \\
\sum\limits_{s \in {{\mathcal S}^u}} {\sum\limits_{i \in {\mathcal I}_s^u} {\left[ {\frac{{\hat \rho \eta }}{M}{\rm tr}({{\bm G}_{i,sm}}) + \sum\limits_{j \in {\mathcal J}} {{{\bar \mu }_{jm}}{\rm{tr}}({\bm Z_j}{{\bm G}_{i,sm}})} } \right]} }
\end{array}
\end{equation}

It is noteworthy that only terms related to power matrices ${\bm V}_{sm}$ and ${\bm G}_{i,sm}$ are involved in (\ref{eq:lagrangian_func}) for brevity as we aim at calculating their ranks via Karush-Kuhn-Tucker (KKT) conditions \cite{Stephen2004Convex}.

By applying KKT conditions, the necessary conditions for achieving the optimal ${\bm V}_{sm}^{\star}$ and ${\bm G}_{i,sm}^{\star}$ can be arranged as
\begin{equation}\label{eq:KKT_condition_1}
    \begin{array}{l}
\frac{{\partial L( \ldots )}}{{\partial {\bm V}_{sm}^{\star}}} =  - \left( {\frac{1}{M} + {{\bar \varphi }_{i,sm}}} \right)\frac{{{\bm H_{i,sm}}}}{{\sigma _{i,s}^2}} + \frac{\eta }{M}{\bm E_{sm}} + \\
\qquad \qquad \sum\limits_{j \in {\mathcal J}} {{{\bar \mu }_{jm}}{\bm Z_j}}  - {{\bm {\Phi}} _{sm}} = \bm 0
\end{array}
\end{equation}
\begin{equation}\label{eq:KKT_condition_2}
    \begin{array}{l}
\frac{{\partial L( \ldots )}}{{\partial {\bm G}_{i,sm}^{\star}}} =  - \left( {\frac{{\hat \rho }}{M} + {{\bar \chi }_{i,sm}}} \right)\frac{{{\bm H_{i,sm}}}}{{\phi \sigma _{i,s}^2}} + \frac{{\hat \rho \eta }}{M}{\bm E_{i,sm}'} + \\
\qquad \qquad \sum\limits_{j \in {\mathcal J}} {{{\bar \mu }_{jm}}{\bm Z_j}}  - {{\bm X}_{i,sm}} = \bm 0
\end{array}
\end{equation}
where $\bm E_{sm}$ and ${\bm E_{i,sm}'}$ are $JK \times JK$ identity matrices.

As the Lagrangian multiplier ${{\bar \mu }_{jm}}$ for all $j \in {\mathcal J}$ is nonnegative, matrices $\frac{\eta }{M}{\bm E_{sm}} + \sum\nolimits_{j \in {\mathcal J}} {{{\bar \mu }_{jm}}{\bm Z_j}}$ and $\frac{{\hat \rho \eta }}{M}{\bm E_{i,sm}'} + \sum\nolimits_{j \in {\mathcal J}} {{{\bar \mu }_{jm}}{\bm Z_j}}$ are full rank. Besides, as ${\bar \varphi}_{i,sm}$ and ${{\bar \chi }_{i,sm}}$ are nonnegative and ${\rm rank}({\bm H}_{i,sm}) = {\rm rank} ({\bm h}_{i,sm}{\bm h}_{i,sm}^{\rm H}) \le 1$, we can conclude that ${\rm rank}($ ${\bm {\Phi}} _{sm}) \ge JK -1$ and ${\rm rank}({\bm X}_{i,sm}) \ge JK - 1$.

On the other hand, the optimal ${\bm V}_{sm}^{\star}$ and ${\bm G}_{i,sm}^{\star}$ will always satisfy the following complementary slackness conditions
\begin{equation}\label{eq:complementary_slack_condition_1}
    {{\bm {\Phi}} _{sm}}{\bm V}_{sm}^{\star} = \bm 0,\forall s \in {{\mathcal S}^e}
\end{equation}
\begin{equation}\label{eq:complementary_slack_condition_2}
    {{\bm X}_{i,sm}}{\bm G}_{i,sm}^{\star} = \bm 0,\forall i \in {\mathcal I}_s^u,s \in {{\mathcal S}^u}
\end{equation}

Since all matrices ${\bm {\Phi}} _{sm}$, ${\bm V}_{sm}^{\star}$, ${{\bm X}_{i,sm}}$, and $\bm G_{i,sm}^{\star}$ are of size of $JK \times JK$, we have ${\rm rank}({\bm {\Phi}} _{sm}) + {\rm rank}({\bm V}_{sm}^{\star}) \le JK$ and ${\rm rank}({{\bm X}_{i,sm}}) + {\rm rank}(\bm G_{i,sm}^{\star}) \le JK$ according to the property of the rank of a matrix. To this end, we obtain the conclusion that ${\rm rank}({\bm V}_{sm}^{\star}) \le 1$ and ${\rm rank}({\bm G}_{i,sm}^{\star}) \le 1$.

Besides, recall that (\ref{eq:standard_transformed_bandwidth_n_beamforming}) is an SDP problem, we may say that the optimal solutions ${\bm V}_{sm}^{\star}$ and ${\bm G}_{i,sm}^{\star}$ to (\ref{eq:standard_transformed_bandwidth_n_beamforming}) can be obtained by some methods such as interior-point methods. This completes the proof.

\subsection{Proof of \textbf{Lemma} 6}
For all $m \in {\mathcal M}$, to prove that ${\bar {\mathcal L}(\omega _{sm}^{e(k)},{\bm v}_{sm}^{(k)},{\bm g}_{i,sm}^{(k)})} $ is bounded, we should prove that variables ${\omega_{sm}^e}$, ${\omega_s^e}$ and ${\psi _{sm}}$ are bounded. Next, we should prove that there exist non-positive coefficients $a_{sm}$ and $a_s$ such that $|{\bar {\mathcal L}^{(k+1)}}  -  {\bar {\mathcal L}^{(k)}}| \le \sum\limits_{s\in {\mathcal S}^e} {a_{sm}}| {\omega _{sm}^{e(k + 1)}}$ $- \omega _{sm}^{e(k)} |  + \sum\limits_{s\in {\mathcal S}^e} {{a_s}| {\omega _s^{e(k + 1)} - \omega _s^{e(k)}} |} $. Since the distributed algorithm is based on the ADMM, we omit the detail proof which is able to be found in the convergence proof of ADMM in \cite{boyd2011distributed,wang2019global}. This completes the proof.

\bibliographystyle{IEEEtran}
\bibliography{network_slice}

\end{document}